\documentclass[lettersize,journal]{IEEEtran}
\usepackage{algorithm}
\usepackage{amsmath}
\usepackage{amsfonts}
\usepackage{array}
\usepackage[caption=false,font=normalsize,labelfont=sf,textfont=sf]{subfig}
\usepackage{textcomp}
\usepackage{stfloats}
\usepackage{url}
\usepackage{verbatim}
\usepackage{microtype}
\usepackage{graphicx}
\usepackage{multirow}
\usepackage{cite}
\usepackage{booktabs}
\usepackage{threeparttable}
\usepackage{pifont}
\usepackage{blindtext}
\usepackage{hyperref}
\usepackage{diagbox}
\usepackage{bbding}
\usepackage{amsthm}
\usepackage{makecell}
\usepackage{algpseudocode}
\usepackage{amssymb}
\usepackage{bm}
\usepackage{pifont}
\usepackage{mathtools}
\usepackage{tabularx}
\usepackage[table]{xcolor}
\definecolor{customblue}{HTML}{BDD7EE}
\definecolor{deepblue}{HTML}{0070C0}
\usepackage{colortbl}

\usepackage[capitalize,noabbrev]{cleveref}
\usepackage{tikz}

\theoremstyle{plain}
\newtheorem{theorem}{Theorem}

\theoremstyle{definition}

\newtheorem{assumption}{Assumption}
\theoremstyle{remark}
\newtheorem{remark}{Remark}

\newtheorem*{proposition*}{Proposition}

\hypersetup{
  colorlinks   = true,    
  urlcolor     = deepblue,    
  linkcolor    = deepblue,    
  citecolor    = deepblue      
}

\newcommand{\name}{\texttt{EGGV}\xspace}

\begin{document}

\title{The Gradient Puppeteer: Adversarial Domination in Gradient Leakage Attacks through Model Poisoning}

\author{Kunlan~Xiang,
        Haomiao~Yang,~\IEEEmembership{Senior Member,~IEEE,}
        Meng Hao,~\IEEEmembership{Member,~IEEE,}
        Shaofeng Li,~\IEEEmembership{Member,~IEEE,}
        Haoxin Wang,
        Zikang Ding,
        Wenbo Jiang,~\IEEEmembership{Member,~IEEE,}
        Tianwei~Zhang,~\IEEEmembership{Member,~IEEE,}

    \IEEEcompsocitemizethanks{
        \IEEEcompsocthanksitem K. Xiang, H. Yang, Z. Ding, and W. Jiang
        are with the School of Computer Science and Engineering, University of Electronic Science and Technology of China. E-mail: klxiang@std.uestc.edu.cn; haomyang@uestc.edu.cn; dzkang0312@163.com; wenbo\_jiang@uestc.edu.cn.
        \IEEEcompsocthanksitem M. Hao is a research scientist at Singapore Management University. E-mail: menghao303@gmail.com.
        \IEEEcompsocthanksitem S. Li is an assistant professor at the School of Computer Science and Engineering at Southeast University. E-mail: shaofengli2013@gmail.com.
        \IEEEcompsocthanksitem H. Wang is with the Sichuan University. E-mail: whx1122@stu.scu.edu.cn.
        \IEEEcompsocthanksitem T. Zhang is with the School of Computer Science and Engineering, Nanyang Technological University, Singapore 639798. E-mail: tianwei.zhang@ntu.edu.sg.
        \IEEEcompsocthanksitem Corresponding author: Haomiao Yang.
        \IEEEcompsocthanksitem This work has been submitted to the IEEE for possible publication. Copyright may be transferred without notice, after which this version may no longer be accessible.
    }      
}




\maketitle

\begin{abstract}
In Federated Learning (FL), clients share gradients with a central server while keeping their data local. However, malicious servers could deliberately manipulate the models to reconstruct clients' data from shared gradients, posing significant privacy risks. Although such active gradient leakage attacks (AGLAs) have been widely studied, they suffer from two severe limitations: (i) coverage: no existing AGLAs can reconstruct all samples in a batch from the shared gradients; (ii) stealthiness: no existing AGLAs can evade principled checks of clients. In this paper, we address these limitations with two core contributions. First, we introduce a new theoretical analysis approach, which uniformly models AGLAs as backdoor poisoning. This analysis approach reveals that the core principle of AGLAs is to bias the gradient space to prioritize the reconstruction of a small subset of samples while sacrificing the majority, which theoretically explains the above limitations of existing AGLAs. Second, we propose \underline{E}nhanced \underline{G}radient \underline{G}lobal \underline{V}ulnerability (EGGV), the first AGLA that achieves complete attack coverage while evading client-side detection. In particular, EGGV employs a gradient projector and a jointly optimized discriminator to assess gradient vulnerability, steering the gradient space toward the point most prone to data leakage. Extensive experiments show that EGGV achieves complete attack coverage and surpasses state-of-the-art (SOTA) with at least a 43\% increase in reconstruction quality (PSNR) and a 45\% improvement in stealthiness (D-SNR).
\end{abstract}

\section{Introduction}

Federated Learning (FL) \cite{mcmahan2017communication, Towardsk, chilimbi2014project} has emerged as a promising framework for privacy-preserving distributed learning, allowing multiple clients to jointly train a global model without sharing raw data. In each communication round, the server distributes the model to clients, who compute gradients on their private data and return them to the server for aggregation. However, a growing body of research has revealed that these shared gradients can be exploited by adversaries to reconstruct private clients' data, an attack known as Gradient Leakage Attacks (GLAs) \cite{zhu2019deep}. GLAs severely compromise the core privacy promise of FL and are broadly categorized into two types \cite{nowak2024qbi}: Passive Gradient Leakage Attacks (PGLAs) \cite{RGAP, yang2022using, 285479}, and Active Gradient Leakage Attacks (AGLAs) \cite{boenisch2023curious, zhao2023loki, fowl2021robbing, nowak2024qbi}.



In PGLAs, attackers reconstruct client data from the gradients shared in the FL system, without manipulating the model architecture, model parameters, and the FL protocol. Zhu et al. \cite{zhu2019deep} first demonstrate the possibility of reconstructing data by optimizing randomly initialized pixel data to produce gradients that match the observed ones. Subsequent methods such as iDLG \cite{zhao2020idlg}, IG \cite{geiping2020inverting}, and STG \cite{yin2021see} improve the reconstruction results by adding image priors as an optimization objective.
However, the effectiveness of these attacks heavily depends on the initialization of model parameters. The model parameters in an unfavorable position produce the gradients lacking data features, rendering these attacks ineffective \cite{du2024sok}. For instance, when model parameters are initialized to zero, the gradients will also be zero and thus contain no data feature, thereby preventing any successful reconstruction by existing GLAs. Our experimental results in Table \ref{compare2} and Figure \ref{visualresults} demonstrate that, for the first time, even previously considered effective PGLAs fail to reconstruct the data when improper initialization methods are used by the server. 
\begin{figure}[t]
    \centering
    \vspace{0.1cm}  
    \includegraphics[scale=0.07]{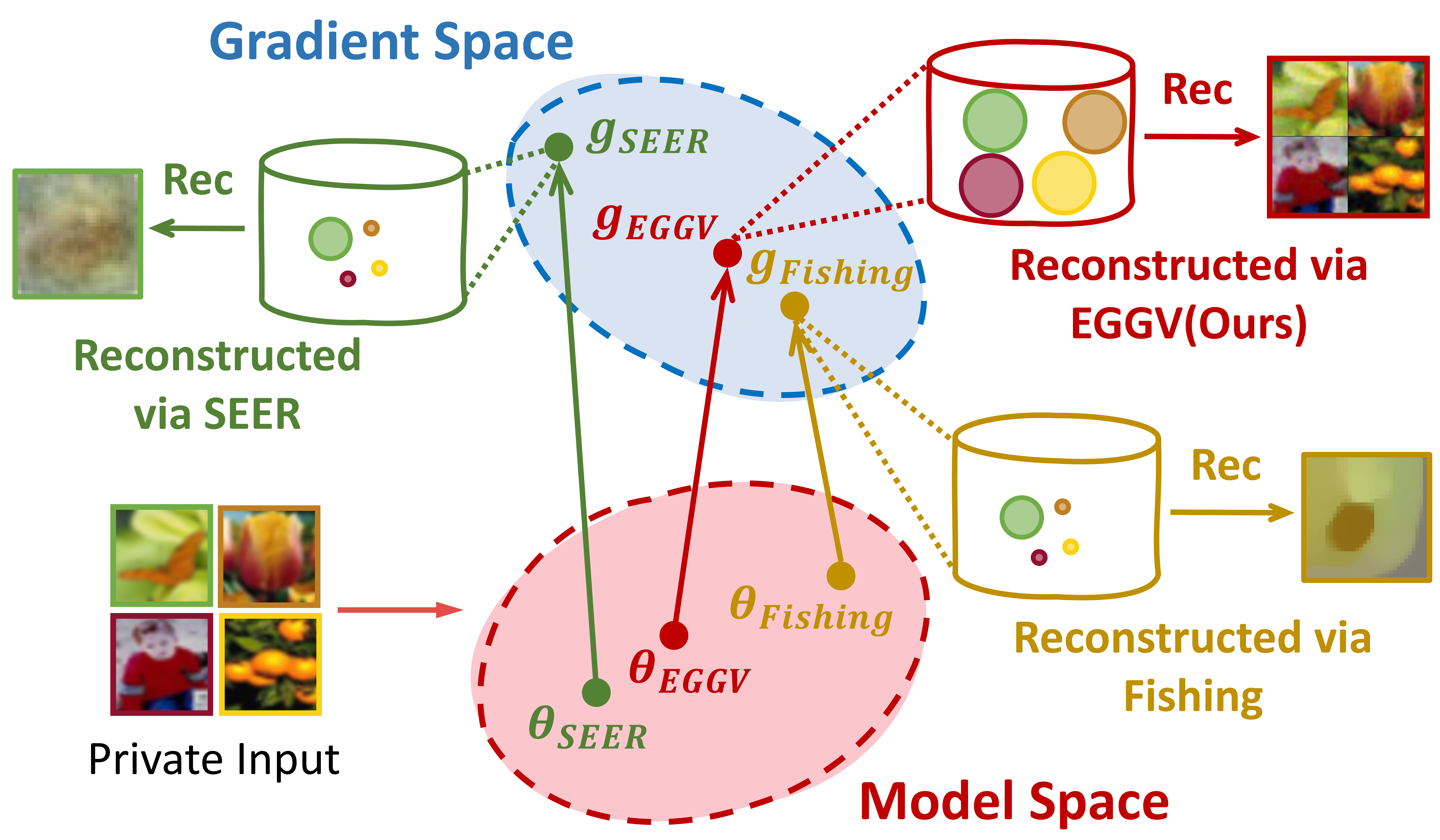}
    \caption{Illustration of the core principles and reconstruction outcomes of Fishing \cite{wen2022fishing}, SEER \cite{Garov2024Hiding}, and \name (Ours). Each cylinder denotes the gradient space of a batch, with inner circle sizes indicating per-sample information retention. Existing AGLAs amplify a few samples while suppressing others, leading to partial recovery. In contrast, \name (Ours) uniformly enhances the feature of all samples in the gradient, enabling full-batch reconstruction.
    }
    \label{showdiff}
\end{figure}

In contrast, in AGLAs, attackers achieve data reconstruction by modifying the model structure and parameters. Some AGLAs \cite{boenisch2023curious, zhao2023loki, fowl2021robbing, nowak2024qbi} achieve direct and accurate reconstruction by inserting a fully connected layer at the beginning of the model and modifying the parameters of this layer. However, such structural modifications are easily detected. Other AGLAs \cite{wen2022fishing, Garov2024Hiding} improve PGLAs by reducing the effective samples used for gradient computation. However, this reduction also limits the improved reconstruction of PGLAs to only a few samples or even just a single sample, as shown in Figure \ref{showdiff}. Moreover, recent work \cite{Garov2024Hiding} reveals that all prior AGLAs are detectable by detection metric. In summary, existing AGLAs exhibit several limitations in attack coverage and stealthiness.




This paper addresses the above critical challenges with two major contributions. First, we propose a new theoretical approach to rethinking and analyzing AGLAs. The core of our approach is the introduction of a parameter $\lambda$, which can quantify the relative contribution of each sample within a batch to the activation of neurons in each class. It discloses that the fundamental principle of existing AGLAs is to prioritize the reconstruction of a small subset of samples with specific properties, while sacrificing the majority of other samples. Such properties are analogous to triggers in backdoor attacks. This principle explains the critical limitations of existing AGLAs, underscoring the pressing need for an advanced attack with a fundamentally different principle. 


Second, based on the above theoretical analysis, we propose \underline{E}nhanced \underline{G}radient \underline{G}lobal \underline{V}ulnerability (\name), a novel AGLA that ensures complete attack coverage and evades client-side detection. 
Different from existing attacks, \name equally enhances the gradient vulnerability of all samples in a batch, as illustrated in Figure \ref{showdiff}. Its key insights include: (i) treating gradients as a latent space of data, with the forward and backward propagation of the model as an encoding process; (ii) introducing a discriminator jointly trained with the model to access the gradient vulnerability of the model. Importantly, \name opens up a new research path thoroughly different from gradient-biased AGLAs.

\textbf{Our contributions are summarized as:} 
\begin{itemize}

\item We introduce a backdoor-theoretic perspective to frame the fundamental principles of AGLAs and identify two critical limitations in their principles: incomplete attack coverage and poor stealthiness.

\item We propose a novel GLA \name that extends AGLAs to achieve both complete coverage and enhanced stealthiness. We further provide theoretical guarantees for the existence of optimal poisoned parameters for gradient leakage and the stealthiness of the proposed attack.

\item Extensive experiments show that the proposed \name significantly surpasses SOTA methods, achieving at least a 43\% improvement in reconstruction quality (measured by PSNR) and a 45\% enhancement in stealthiness (measured by D-SNR) with a complete attack coverage.

\end{itemize}

\textbf{Paper Organization.} The remainder of this paper is organized as follows. Section~\ref{sec:related-work} reviews related work on GLAs. Section~\ref{sec:background} provides the necessary background, including the FL framework, the threat model, and motivation for this work. In Section~\ref{sec:backdoor-analysis}, we present a backdoor-theoretic analysis of AGLAs to reveal the fundamental limitations in their principles. Section~\ref{sec:EGGV} introduces the proposed attack, detailing its formulation, optimization approach, and theoretical guarantees. Experimental results are presented in Section~\ref{sec:experiments}. Finally, Section~\ref{sec:conclusion} draw a conclusion.

\section{Related Work}
\label{sec:related-work}
In this section, we provide an overview of prior GLAs.

\subsection{Passive Gradient Leakage Attack (PGLA)}
Most PGLAs are carried out through gradient matching. The DLG attack \cite{zhu2019deep} is the first to optimize dummy inputs and corresponding dummy labels by matching their gradients to the observed gradients. Its optimization objective is: 
\begin{equation}
\label{dlg}
    x'^*,y'^*=\underset{x',y'}{\operatorname*{argmin}}\left\|\frac{\partial\ell(F(x',\theta),y')}{\partial\theta}-\nabla\theta\right\|^2,
\end{equation}
DLG works well on small models, such as LeNet-5 \cite{lecun1998gradient}. Later, iDLG \cite{zhao2020idlg} further improves DLG by inferring data labels from the gradients, but this method is limited to batches where each sample has a unique label. Subsequent works, such as LLG \cite{wainakh2022user}, extend iDLG to handle larger batch sizes, while other methods \cite{ma2023instance, wangtowards}, such as instance-wise reconstruction \cite{ma2023instance}, successfully recover ground-truth labels even in large batches with duplicate labels. IG \cite{geiping2020inverting} introduces Total Variation and a Regularization term to improve the optimization objective further. The study \cite{yin2021see} leverages the mean and variance from batch normalization layers as priors to enhance GLAs. Instead of optimizing dummy inputs directly, GI \cite{jeon2021gradient} proposes optimizing the generator and its input latent to generate dummy images whose gradients match the observed gradients. GGL \cite{li2022auditing} simplifies this by focusing solely on the latent space of pre-trained BigGAN for gradient matching. More recently, GGDM \cite{gu2024federated} uses captured gradients to guide a diffusion model for reconstruction. However, despite these advancements, PGLAs fail under improper model parameter initializations that yield gradients containing minimal data features. Our experiments in Table \ref{compare2} and Figure \ref{visualresults}, for the first time, question the practical effectiveness of these attacks under popular model initialization methods. This exposes a critical limitation in real-world scenarios.

\subsection{Active Gradient Leakage Attack (AGLA)}
Based on the manipulation strategies, most AGLAs can be classified into two categories: structure-modified AGLAs and gradient-biased AGLAs.

\textbf{Structure-Modified AGLAs}.
This type of attack mainly enhances PGLAs by assuming a dishonest server that manipulates the model structure \cite{boenisch2023curious, fowl2021robbing, nowak2024qbi, zhao2023loki}. Some studies \cite{fowl2021robbing, nowak2024qbi} insert an FC layer at the beginning of the model, while another work \cite{zhao2023loki} inserts a convolutional layer and two FC layers. These inserted layers are referred to as ``trap weights", which are maliciously modified so that the neurons inside are activated only by samples with specific properties, enabling the reconstruction of the sample with the strongest property. However, modifications to the model structure are inherently detectable due to their explicit changes to the model architecture, rendering them impractical in real-world scenarios. Therefore, this work focuses on another type of AGLAs \cite{zhang2022compromise, wen2022fishing, Garov2024Hiding} which poisons model parameters instead of modifying the model structure. 

\textbf{Gradient-Biased AGLAs}.
Gradient-biased AGLAs \cite{zhang2022compromise, pasquini2022eluding, wen2022fishing, Garov2024Hiding} poison model parameters to skew the gradient space, ensuring that selected samples dominate the batch-averaged gradients while suppressing others. For instance, the attack \cite{zhang2022compromise} zeros out most of the convolutional layers, ensuring that only one sample's features reach the classification layer, activating the relevant neurons.  The method \cite{wen2022fishing} assigns many 0s and 1s to the last FC layer to make the averaged gradient close to the gradient of a single sample, thereby enhancing the reconstruction of PGLAs on a single sample. The method \cite{pasquini2022eluding} distributes inconsistent models to clients, forcing non-target users’ ReLU layers \cite{nair2010rectified} to output zero gradients, thereby retaining only the target user’s gradients, which can then be exploited to leak the targeted private data. Recent work \cite{Garov2024Hiding} observes that all these AGLAs bias the averaged gradient toward the gradients of a small subset of data within a batch while suppressing the gradients of other samples. Exploiting this bias in the biased gradients, research \cite{Garov2024Hiding} introduces a D-SNR detection metric to check poisoned model parameters, which is calculated as below:
\begin{equation}
\label{dsnr}
\resizebox{\linewidth}{!}{$
\begin{aligned}
    D-SNR(\boldsymbol{\theta})= \max_{W\in\boldsymbol{\theta}_{lw}}\frac{\max_{i\in\{1,...,B\}}\left\|\frac{\partial\ell(F(x_i),y_i)}{\partial W}\right\|}{\sum_{i=1}^B\left\|\frac{\partial\ell(F(x_i),y_i)}{\partial W}\right\|-\max_{i\in\{1,...,B\}}\left\|\frac{\partial\ell(F(x_i),y_i)}{\partial W}\right\|},
\end{aligned}
$}
\end{equation}
where ${\theta}_{lw}$ denotes the set of weights of all dense and convolutional layers. D-SNR claims that all prior AGLAs are detectable by principled checks.

\section{Background}
\label{sec:background}
This section provides the necessary background for understanding our work. We first present an overview of the FL paradigm. We then describe the adversarial assumptions and capabilities underlying our threat model. We also discuss the motivation behind our proposed attack. For clarity, the main notations used in this paper are summarized in Table~\ref{tab:notation}.
\begin{table}[t]
\centering
\caption{Major notations used in this paper.}
\label{tab:notation}
\begin{tabular}{ll}
\toprule
Notation & Description \\
\midrule
$x$ & Data batch input into the model \\
$x_i$ & $i$-th sample in a batch $x$ \\
$y$, $y_i$ & Ground-truth label for $x$ or $x_i$ \\
$\hat{y}$ & Model prediction, i.e., $\hat{y}=F(x,\theta)$ \\
$F(\cdot)$ & Neural network model \\
$\theta$ & Model parameters \\
$W$, $W^k$ & Weight matrix and its $k$-th column in an FC layer \\
$\eta$ & Learning rate \\
$\mathcal{D}(\cdot)$ & Gradient-to-input decoder network \\
$\ell(\cdot, \cdot)$ & Loss function (e.g., cross-entropy) \\
$\nabla_\theta \ell$ & Gradient of the loss with respect to $\theta$ \\
$R(\cdot)$ & Gradient leakage reconstruction function \\
$x'$ & Reconstructed input from gradients \\
$B$ & Batch size \\
$C$ & The number of classification classes \\
$k$ & Class index \\
$\theta^*$ & Optimized model parameters maximizing gradient leakage \\
$\phi$ & Parameters of the decoder $\mathcal{D}$ \\
$\lambda_i^k$ & Contribution weight of $x_i$ to class-$k$ neuron gradient \\
$\mathbf{\Lambda}$ & Weight matrix formed by $\lambda_i^k$ over $i$ and $k$ \\
$\bar{x}^{(k)}$ & Weighted average input reconstructed via class-$k$ gradient \\
$b$, $b^k$ & Bias vector and its $k$-th element \\
$\Pi(\cdot)$ & Gradient projection operator \\
$\tilde{g}$ & Projected gradient vector \\
$\rho$ & Projection ratio of gradient dimensionality \\
$L(\theta, \phi)$ & \name training loss for model and decoder \\
$D$ / $D_a$ & Client dataset / Auxiliary dataset available to attacker \\
\bottomrule
\end{tabular}
\end{table}

\subsection{Federated Learning}
Federated Learning (FL) is a decentralized learning framework that enables multiple clients to collaboratively train a global model without exchanging their raw data, thereby preserving data privacy~\cite{mcmahan2017communication}. In each communication round $t$, a central server selects a subset of clients $\mathcal{C} = \{c_1, c_2, \dots, c_n\}$ and distributes the current global model $F(\theta_t)$ to them. Each client $c_i$ performs local training on its private dataset $D_i$ by minimizing the empirical loss $\ell(F(\theta_t), D_i)$ and computes the corresponding gradient:
\begin{equation}
    g_i^t = \nabla_{\theta_t} \ell(F(\theta_t), D_i).
\end{equation}
Clients then upload their gradients $\{g_i^t\}_{i=1}^n$ to the server. The server aggregates the gradients (e.g., via weighted averaging) and updates the global model as follows:
\begin{equation}
    \theta_{t+1} = \theta_t - \eta \sum_{i=1}^n w_i \cdot g_i^t, \quad \text{where } w_i = \frac{|D_i|}{\sum_j |D_j|},
\end{equation}
where $\eta$ denotes the learning rate. This process is repeated iteratively until model convergence. Although FL avoids direct access to raw data, recent studies have shown that the shared gradients can still reveal client data. This vulnerability has led to a line of research on GLAs, which aim to reconstruct private data from gradients. Our work builds upon this threat by exploring more general and stealthy attack mechanisms.

\subsection{Threat Model}
Our threat model operates within an FL framework where the server is dishonest and curious. It attempts to infer private client data by poisoning the global model parameters before distribution. However, the server is constrained from modifying the model architecture, as structural changes (e.g., inserting layers) are easily detected by clients via architecture verification, integrity checks, or test queries. Following the setting in many prior works \cite{Garov2024Hiding, yang2022using, 285479}, we assume the malicious client can take some publicly available datasets as the auxiliary dataset, which can be easily obtained from open repositories such as Hugging Face\cite{huggingface}, Kaggle\cite{kaggle}, or OpenML\cite{OpenML}. This is a realistic assumption, as adversaries can readily download such datasets in practical scenarios. Moreover, we experimentally demonstrate that the effectiveness of our attack is not sensitive to the distribution gap between the auxiliary and target datasets, further supporting the rationality of this setting.

\subsection{Motivation}
While PGLAs have been extensively studied and shown to be effective in controlled experimental settings, their success is highly sensitive to the model's parameter initialization. Our empirical observations reveal that models initialized using commonly adopted schemes such as Random, Xavier \cite{glorot2010understanding}, and He \cite{he2015delving} may fail to embed enough data features into the gradients, rendering even SOTA PGLAs incapable of reconstructing the original samples. This strong dependence on the model parameters fundamentally limits PGLAs to a passive role—attackers lack control over the feature distribution in the gradient space, leading to unpredictable and unreliable performance in real-world deployments.

AGLAs have attempted to overcome these limitations by modifying the model architecture or injecting targeted biases into the parameter space. However, such approaches often amplify the gradient of only a subset of samples while suppressing others, resulting in low attack coverage. Moreover, these manipulations typically introduce detectable anomalies in the gradients or model behavior, thereby increasing the risk of being discovered by defensive mechanisms on the client side. This indicates that existing AGLAs are inherently incapable of simultaneously ensuring attack stealth and integrity, due to their fundamental principles.

Inspired by these observations, we propose a novel attack: enhancing the overall information capacity of the gradient space through balanced strengthening. We note that if we can actively guide the model to learn a ``balanced encoding" gradient space during training, ensuring that all samples have similar information expression strengths in the gradient (i.e., enhancing the leakage potential of each sample), this would not only guarantee the reconstruction ability of all samples but also maintain the naturalness of the gradient distribution, reducing the risk of detection. Thus, our method not only breaks free from the traditional passive reliance of PGLAs on waiting for data characteristics to leak but also avoids the incomplete reconstruction coverage and detectability issues caused by the ``intentional bias" in existing AGLAs. This provides a more stable, comprehensive, and stealthy GLA.



\section{Backdoor-theoretical Analysis}
\label{sec:backdoor-analysis}
We introduce a new approach for AGLA analysis. It
offers a deep insight into the relationship between model parameters and gradient bias, and explains why existing AGLAs are detectable and cannot recover all samples in a batch. 

For a simple neural network that is only comprised of fully connected layers $F(x)=xW+b$, where $x\in\mathbb{R}^{B\times m}$ is a batch of data, $W\in\mathbb{R}^{m \times n}$ is the weight parameters. $b\in\mathbb{R}^{1 \times n}$ is the bias, with $B$ being the batch size and $n$ being the number of classification categories. When data $x$ is fed into the model, the output is represented by $\hat{y}=xW+b$. 
As seen in prior work \cite{fowl2021robbing}, the gradients of the weights and biases of the FC layer can be directly used to reconstruct a weighted average of the input data:
\begin{equation}
\label{avgx}
    \bar{x}^{(k)}=\frac{\nabla_{W^k}\ell(F(x,\theta),y)}{\nabla_{b^k}\ell(F(x,\theta),y)}=\sum_{i=1}^B\lambda_i^k\cdot x_i,
\end{equation}
where $k\in[1,n]$ is the class index, $\nabla_{W^k}\ell(F(x,\theta),y)$ (abbreviated as $\nabla W^{k}$) denotes the gradient of the $k^{th}$ column of the weight matrix $W$, and $\nabla_{b^k}\ell(F(x,\theta),y)$ (abbreviated as $\nabla b^{k}$) represents the gradient of the $k^{th}$ element of the bias $b$. 
$\lambda$ is defined as below.

\begin{theorem}
\label{lambda}
    Let $F(x)=xW+b$ be the classification model with one FC layer, where $x$ is the input data, $W$ is the weight matrix, and $b$ is the bias vector, and the corresponding model output is $\hat{y}=F(x)$. Suppose $\ell(\hat{y}, y)$ is the loss function between the model output $\hat{y}$ and the ground-truth labels $y$. For any class index $k\in\{1, 2, ..., C\}$ and sample index $i\in\{1, 2, ..., B\}$, the coefficient $\lambda$ holds that:
    \begin{equation}
    \lambda_i^k=\frac{\frac{\partial \ell(\hat{y}_i^k,y_i)}{\partial\hat{y}_i^k}}{\sum_{j=1}^B\frac{\partial \ell(\hat{y}_j^k,y_j)}{\partial\hat{y}_j^k}}, and \sum_{i=1}^B\lambda_i^k=1,
    \end{equation}
    where ${\partial \ell(\hat{y}_i^k,y_i)}/{\partial\hat{y}_i^k}$ denotes the partial derivative of the loss function with respect to the output $\hat{y}_i^k$.
\end{theorem}
\begin{proof}
    Consider a batch of $B$ samples $\{(x_i,y_i)\}_{i=1}^B$, where $x\in\mathbb{R}^{B\times\text{Channel}\times\text{Height}\times\text{Width}}$ represents the input data and $y_i$ are the corresponding ground-truth labels. The model outputs for each sample within a batch are given by: $\hat{y}_i=x_iW+b\in\mathbb{R}^C$, where $C$ denotes the number of classification classes. The total loss over the batch is $\frac1B\sum_{i=1}^B\ell(\hat{y}_i,y_i)$. Next, we derive the gradients of weights and biases for the $k^{th}$ class and express the ratio $\frac{\nabla W^k}{\nabla b^k}$ in terms of $\lambda_{i}^{k}$ and $x_i$. The gradient of the weight matrix $W$ with respect to the loss for class index $k$ is given by the average of the gradients over all samples:
    \begin{equation}
    \nabla W^k=\frac1B\sum_{i=1}^B\nabla W_i^k.
    \end{equation}
    According to the chain rule, we can further obtain:
    \begin{equation}
    \begin{aligned}
    \nabla W^{k}& =\frac1B\sum_{i=1}^B\nabla W_i^k=\frac1B\sum_{i=1}^B\frac{\partial l(\hat{y}_i^k,y_i)}{\partial\hat{y}_i^k}\cdot\frac{\partial\hat{y}_i^k}{\partial W_i^k} \\
    &=\frac1B\sum_{i=1}^B\frac{\partial l(\hat{y}_i^k,y_i)}{\partial\hat{y}_i^k}\cdot x_i.
    \end{aligned}
    \end{equation}
    Similarly, the gradient of the bias corresponding to the $k^{th}$ category index can be derived as:
    \begin{equation}
    \begin{gathered}
    \nabla b^{k}=\frac1B\sum_{i=1}^B\nabla b_i^k=\frac1B\sum_{i=1}^B\frac{\partial l(\hat{y}_i^k,y_i)}{\partial\hat{y}_i^k}\cdot\frac{\partial\hat{y}_i^k}{\partial b^k} \\
    =\frac1B\sum_{i=1}^B\frac{\partial l(\hat{y}_i^k,y_i)}{\partial\hat{y}_i^k}\cdot1
    =\frac1B\sum_{i=1}^B\frac{\partial l(\hat{y}_i^k,y_i)}{\partial\hat{y}_i^k}.
    \end{gathered}
    \end{equation}
    Therefore, $\nabla W^k/\nabla b^k$ can be derived as:
    \begin{equation}
    \resizebox{\linewidth}{!}{$
    \begin{aligned}
    \frac{\nabla W^k}{\nabla b^k}& =\frac{\frac1B\sum_{i=1}^B\frac{\partial l(\hat{y}_i^k,y_i)}{\partial\hat{y}_i^k}\cdot x_i}{\frac1B\sum_{i=1}^B\frac{\partial l(\hat{y}_i^k,y_i)}{\partial\hat{y}_i^k}}=\frac{\sum_{i=1}^B\frac{\partial l(\hat{y}_i^k,y_i)}{\partial\hat{y}_i^k}\cdot x_i}{\sum_{i=1}^B\frac{\partial l(\hat{y}_i^k,y_i)}{\partial\hat{y}_i^k}} \\
    &=\frac{\frac{\partial l(\hat{y}_{1}^{k},y_{1})}{\partial\hat{y}_{1}^{k}}\cdot x_{1}}{\sum_{i=1}^{B}\frac{\partial l(\hat{y}_{i}^{k},y_{i})}{\partial\hat{y}_{i}^{k}}}+\frac{\frac{\partial l(\hat{y}_{2}^{k},y_{2})}{\partial\hat{y}_{2}^{k}}\cdot x_{2}}{\sum_{i=1}^{B}\frac{\partial l(\hat{y}_{i}^{k},y_{i})}{\partial\hat{y}_{i}^{k}}}+\cdots+\frac{\frac{\partial l(\hat{y}_{B}^{k},y_{B})}{\partial\hat{y}_{B}^{k}}\cdot x_{B}}{\sum_{i=1}^{B}\frac{\partial l(\hat{y}_{i}^{k},y_{i})}{\partial\hat{y}_{i}^{k}}} \\
    &=\sum_{i=1}^{B}\frac{\frac{\partial l(\hat{y}_{i}^{k},y_{i})}{\partial\hat{y}_{i}^{k}}}{\sum_{j=1}^{B}\frac{\partial l(\hat{y}_{j}^{k},y_{j})}{\partial\hat{y}_{j}^{k}}}\cdot x_{i}.
    \end{aligned}
    $}
    \end{equation}
    This expression can be rewritten as:
    \begin{equation}
    \frac{\nabla W^k}{\nabla b^k}=\sum_{i=1}^{B}\frac{\frac{\partial l(\hat{y}_{i}^{k},y_{i})}{\partial\hat{y}_{i}^{k}}}{\sum_{j=1}^{B}\frac{\partial l(\hat{y}_{j}^{k},y_{j})}{\partial\hat{y}_{j}^{k}}}\cdot x_{i}=\sum_{i=1}^B\lambda_i^k\cdot x_i.
    \end{equation}
    Therefore, $\lambda_i^k={\frac{\partial \ell(\hat{y}_i^k,y_i)}{\partial\hat{y}_i^k}}/{\sum_{j=1}^B\frac{\partial \ell(\hat{y}_j^k,y_j)}{\partial\hat{y}_j^k}}, and \sum_{i=1}^B\lambda_i^k=1.$
\end{proof}

Taking a binary classification network with an input of 4 samples as an example, the weighted average sample obtained by the gradient of the weights and biases of the two categories can be expressed as:
\begin{equation}
\label{lambdaexample}
    \begin{bmatrix}\frac{\nabla W^1}{\nabla b^1}\\\frac{\nabla W^2}{\nabla b^2}\end{bmatrix}=\begin{bmatrix}\bar{x}^{(1)}\\\bar{x}^{(2)}\end{bmatrix}=\begin{bmatrix}\lambda_1^1&\lambda_2^1&\lambda_3^1&\lambda_4^1\\\lambda_1^2&\lambda_2^2&\lambda_3^2&\lambda_4^2\end{bmatrix}\begin{bmatrix}x_1\\x_2\\x_3\\x_4\end{bmatrix}=\mathbf{\Lambda} \mathbf{X}.
\end{equation}
Equation (\ref{lambdaexample}) shows the weighted average data resolved by the gradient of the weights and biases $w.r.t$ a given class, which is actually a weighted summation of the features of the input layer. $\lambda$ is exactly the weight factor to quantify the bias in the neurons' gradient space toward specific samples within a batch.

$\lambda$ is crucial for the gradient-biased AGLAs, as it controls the weighted feature proportions computed from gradients. It is the first technical measure to quantify the relative contribution of each sample in the batch to the activation of each class of neurons in the model. Interestingly, we find the AGLAs are very analogous to poisoning-based backdoor attacks. In the later, the attacker poisons the model training process to an expected state to manipulate the model output and eventually control the distribution of $\lambda$. Therefore, we call this approach the backdoor-theoretical perspective.

The core mechanism of manipulating $\lambda$ inherently results in several challenges: (1) samples without gradient space bias cannot be reconstructed; (2) reconstruction becomes impossible when two samples with the required properties coexist; and (3) the presence of anomalous gradients in the gradient space makes the attack detectable. This explains the fundamental limitations of existing AGLAs. 

This is the first theoretical analysis for AGLAs, which systematically reveals the underlying mechanism behind gradient bias via the backdoor-theoretical lens. It not only bridges a critical gap in the current literature but also provides principled insights into why existing AGLAs are inherently incomplete in reconstruction and detectable due to gradient anomalies.

\section{Enhanced Gradient Global Vulnerability}
\label{sec:EGGV}
The above-mentioned challenges suggest that instead of controlling $\lambda$ to achieve reconstruction, we should focus on increasing the concentration of input features and enhancing feature representation at the source rather than compressing the features of some samples to amplify others. Following this inspiration, we introduce \name, a new attack that poisons the model parameters $\theta$ to equally enhance the leakage potential of all samples in a batch, thus ensuring a comprehensive attack and evading detection.
Figure \ref{showchallenges} provides an intuitive comparison between \name and existing Gradient-biased AGLAs. While existing attacks achieve reconstruction by suppressing the features of non-target samples to amplify those of specific ones in the gradient, \name instead uniformly enhances the encoded features of each sample in the gradient, following an entirely different principle.
\begin{figure}[t]
    \centering
    \vspace{0.0cm}  
    \includegraphics[scale=0.1]{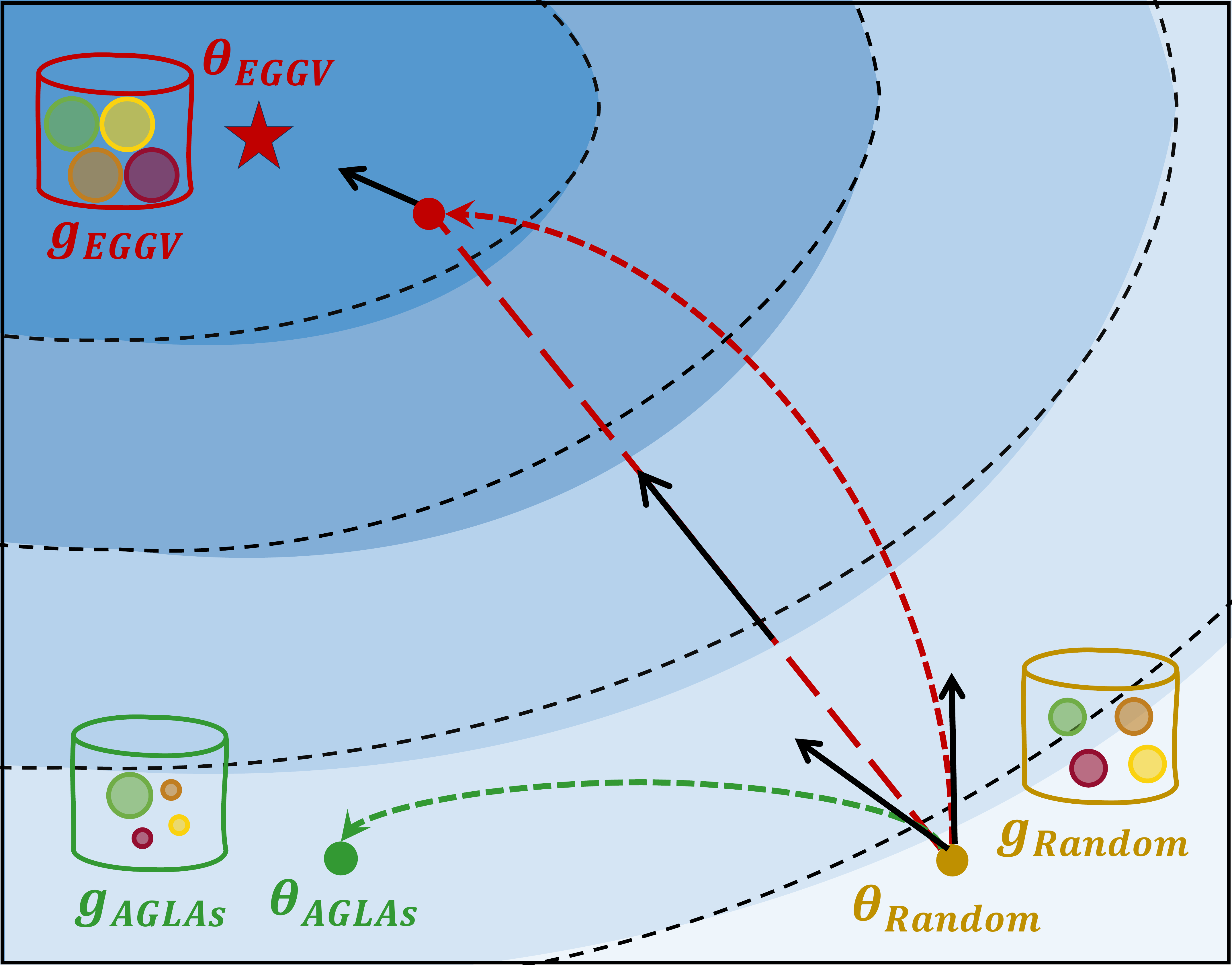}
    \caption{Fundamental principle comparison between \name and gradient-biased AGLAs.}
    \label{showchallenges}
\end{figure}
\begin{figure*}[t]
    \centering
    \vspace{0.0cm}  
    \includegraphics[scale=0.23]{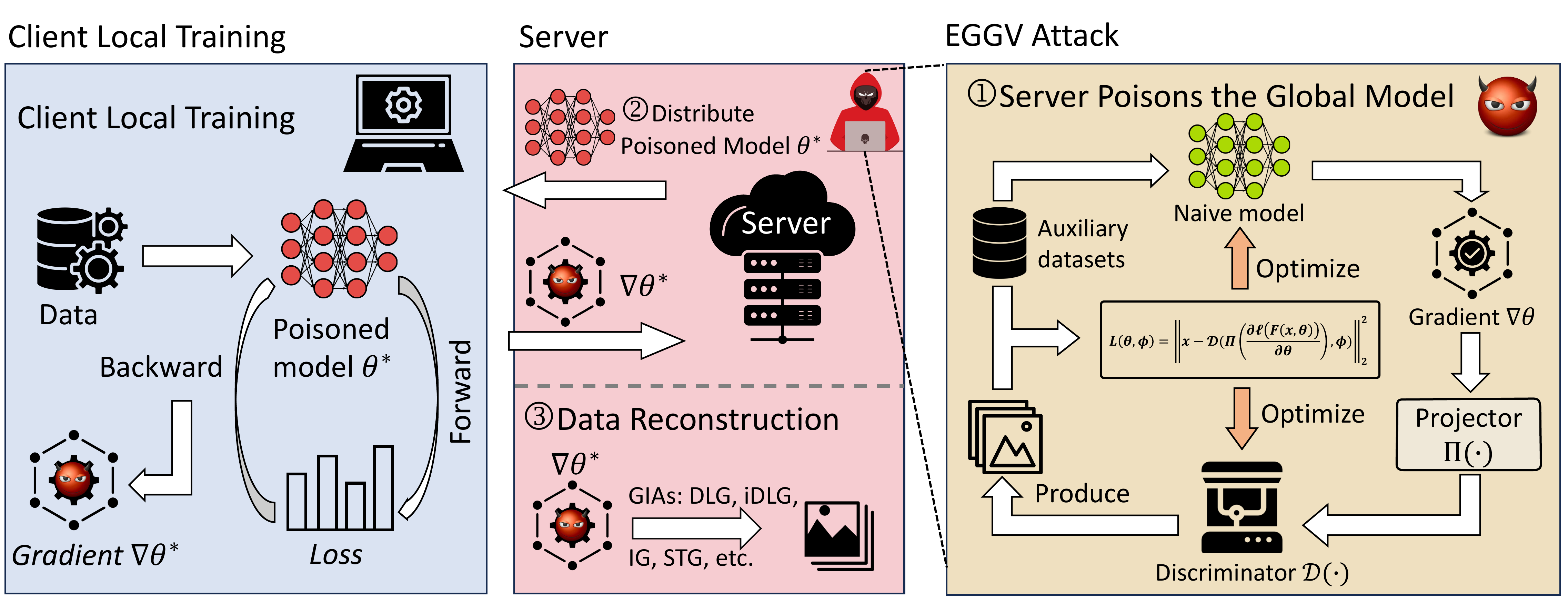}
    \caption{Overview of the proposed \name, consisting of three steps: \ding{172} poison the model parameters to make its gradient space vulnerable; \ding{173} distribute poisoned model and gather vulnerable gradients; \ding{174} implement existing GLAs on these gradients.}
    \label{EGGVoverview}
\end{figure*}
\begin{algorithm}
\caption{Poisoning and Reconstruction of \name}
\label{algorithm1}
\begin{algorithmic}[1]
\State \textbf{Input:} Auxiliary dataset $D_a$, global model $F(\theta)$, acceptable error $\epsilon$, number of iterations $N$
\State \textbf{Output:} Reconstructed data $x^{\prime}$

\State \textbf{Main Process:}
\State $F(\theta^*) \gets$ \texttt{PoisonModel}($D_a, F(\theta), N$)
\State $\nabla \theta^* \gets$ \texttt{CollectClientGradients}($F(\theta^*)$)
\State $x^{\prime} \gets$ \texttt{ReconstructData}($\nabla \theta^*$)
\State \Return $x^{\prime}$

\Function{PoisonModel}{$D_a, F(\theta), N$}
\State Initialize $\theta_{0}$ randomly
\State $t \gets 0$
\While{$L(\theta_{t}, \phi_{t}) < \epsilon$}
    \For{each $(x_j, y_j) \in D_a$}
        \State $L(\theta_{t}, \phi_{t}) \gets \left\|x_j - \mathcal{D}\left(\Pi\left(\frac{\partial \ell(F(x_j, \theta_{t}), y_j)}{\partial \theta_{t}}\right), \phi_{t}\right) \right\|_2^2$
        \State Update $\theta_{t}, \phi_{t}$ using gradient descent:
        \State \quad $\theta_{t+1} \gets \theta_{t} - \alpha_1 \nabla_{\theta_{t}} L(\theta_{t}, \phi_{t})$
        \State \quad $\phi_{t+1} \gets \phi_{t} - \alpha_2 \nabla_{\phi_{t}} L(\theta_{t}, \phi_{t})$
        \State $t \gets t + 1$
    \EndFor
\EndWhile
\State \Return $F(\theta^*)$ with updated $\theta$
\EndFunction

\Function{CollectClientGradients}{$F(\theta^*)$}
    \State Client $i$ receives the global model $F(\theta^*)$
    \State Client $i$ calculates the gradient $\nabla_{\theta^*} \ell(F(x, \theta^*), y)$ using its data $(x, y)$
    \State \Return $\nabla_{\theta^*} \ell(F(x,\theta^*), y)$
\EndFunction

\Function{ReconstructData}{$\nabla \theta^*$}
    \State Select any prior PGLA methods $R(\cdot)$
    \State Reconstruct client data $x^{\prime}$ through $R(\nabla \theta^*)$
    \State \Return reconstructed data $x^{\prime}$
\EndFunction

\end{algorithmic}
\end{algorithm}
\subsection{Attack Overview}
The proposed \name is comprised of three main steps. Algorithm \ref{algorithm1} shows the detailed process. Figure \ref{EGGVoverview}  provides a visual overview of \name.

\textbf{Step I:} The malicious server poisons the global model before its distribution (Step \ding{172} in Figure \ref{EGGVoverview}, \texttt{PoisonModel} in Algorithm \ref{algorithm1}). In this stage, the server iteratively optimizes the global model and discriminator locally with the objective function $L(\theta, \phi)$ using an auxiliary dataset, enriching the gradients with encoded features.

\textbf{Step II:} The server distributes the poisoned model to clients (Step \ding{173} in Figure \ref{EGGVoverview}, \texttt{CollectClientGradients} in Algorithm \ref{algorithm1}). Each client, following the FL protocol, feeds its own training data into the poisoned model to generate gradients that are then uploaded back to the server.

\textbf{Step III:} The server uses these uploaded gradients to perform data reconstruction using any existing PGLAs (Step \ding{174} in Figure \ref{EGGVoverview}, \texttt{ReconstructData} in Algorithm \ref{algorithm1}).

\subsection{Problem Formulation}
The objective of the malicious server is to minimize the difference between the reconstructed data and the original data. Formally, we have the following objective:
\begin{equation}
    \theta^*=\underset{\theta}{\operatorname*{argmin}}\|x-x^{\prime}\|_p,
\end{equation}
where $x$ represents one data batch from the auxiliary dataset $D$, and $x^{\prime}$ denotes the reconstructed data obtained by the attacker using the gradient leakage method $R(\cdot)$. These gradients are computed on the client side during local training. 
Specifically, the client inputs local training data $x$ into the model, yielding the output $\hat{y}=F(x,\theta)$, 
and then calculates the gradient $\frac{\partial \ell(F(x,\theta),y)}{\partial\theta}$. The optimization objective can therefore be expanded as follows:
\begin{equation}
\label{poisonform}
    \theta^*=\arg\min_\theta\|x-R\left(\frac{\partial \ell(F(x,\theta),y)}{\partial\theta}\right)\|_p.
\end{equation}
The server aims to optimize the model parameters $\theta$ such that the reconstructed data obtained via the gradient leakage method $R(\cdot)$, closely approaches the original input data $x$.

\subsection{Detailed Solution}
As shown in Equation~\eqref{poisonform}, the performance of the gradient leakage attack depends on the model parameters, meaning there exists an optimal set of model parameters, denoted as $\theta^\ast$, that minimizes the reconstruction loss between the original data $x$ and the recovered data $x'$ given a specific reconstruction function $R(\cdot)$. 

To search in the continuous parameter space for the optimal model parameters, 
we introduce a dimension reduction projector $\prod(\cdot)$. 
Specifically, for a data batch $x$, the gradients $g=\nabla_\theta\ell(F(x,\theta),y)$ on the global model are sampled by the projector $\prod(\cdot)$ at fixed positions, ensuring that the gradients have consistent positions during the iterations:
\begin{equation}
    \Pi(g)=(g_{1}[p_1],...,g_{L}[p_L], \rho)^T,
\end{equation}
where $p_1$, $\ldots$, $p_L$ represent pre-specified sets of gradient positions from the $1^{st}$ to the $L^{th}$ layer, indicating the fixed positions of the gradients sampled during each iteration. $\rho$ represents the ratio of the number of parameters of the projected gradient to that of the original gradient, and we hereinafter refer to it as the projection ratio. By applying the projector $\Pi(\cdot)$, we map the high-dimensional gradient space into a lower-dimensional vector space as follows:
\begin{equation}
    \tilde{g}=\Pi\left(\frac{\partial \ell(F(x,\theta),y)}{\partial\theta}\right).
\end{equation}
With the projected gradient, we introduce a discriminator $\mathcal{D}(\cdot)$ to evaluate the potential for gradient leakage by the projected gradient. 
Equation~\eqref{poisonform} measures the vulnerability of the corresponding gradient space by performing an end-to-end reconstruction attack to compute the similarity between the reconstructed and original data. 
Traditional end-to-end reconstruction attacks require performing gradient matching in Equation (\ref{dlg}) and then optimizing on dummy data. This method has a high computational cost, as it requires iterative operations in the continuous gradient space for each update of $\theta$. Additionally, iterative reconstruction introduces a nested optimization structure, making it challenging to compute the second-order derivative of the loss.

Prior GLAs suggest that the gradient vulnerability stems from the data features they encode. The more data features a gradient contains, the more vulnerable that gradient becomes. To this end, \textit{the gradients can be encoded representations of input data}, while the forward and backward propagation within the model is the encoding process.

The goal of the proposed \name is to refine this encoding process to maximize the retention of input data features within the gradients. 
To achieve this, we design a decoder that decodes the projected gradients to the original input. 
The attacker jointly optimizes the model and the decoder by the following loss function:
\begin{equation} \label{eq:decoder}
    L(\theta, \phi) = \left\|x - \mathcal{D}\left(\Pi\left(\frac{\partial \ell(F(x, \theta), y)}{\partial \theta}\right), \phi\right) \right\|_2^2,
\end{equation}
where $\phi$ represents the parameters of the decoder $\mathcal{D}$.

\subsection{Optimal Model for the Gradient Leakage}
We prove a global minimum exists for the proposed loss function, where the corresponding model parameters maximize the vulnerability of gradient space to data leakage.
\begin{assumption}
\label{assumption1}
The global model $F(\theta)$ is continuous, and the parameter space $\Theta$ of $\theta$ is a non-empty compact set.
\end{assumption}
\begin{assumption}
\label{assumption2}
The loss function $\ell(\cdot, \cdot)$ for the client training is continuously differentiable with respect to the model parameters $\theta$, allowing for gradient computation with respect to $\theta$.
\end{assumption}
\begin{assumption}
\label{assumption3}
The decoder $\mathcal{D}$ is a linear function of the form $\mathcal{D}(\tilde{g})=W\cdot\tilde{g}+b$ and is continuous. The parameter space $\Phi$ of $\phi$ is a non-empty compact set.
\end{assumption}
\begin{theorem}
\label{theorem1}
Under the above assumptions, there exists parameters $\theta^*\in\Theta$, $\phi^*\in\Phi$ such that the loss function $L(\theta, \phi)$ defined in Equation~\eqref{eq:decoder} attains its global minimum:
\begin{equation}
    \theta^*, \phi^*=\arg\min_{\theta\in\Theta, \phi\in\Phi}L(\theta, \phi).
\end{equation}
\end{theorem}
At $\theta=\theta^{*}, \phi=\phi^{*}$, the gradient $\nabla_\theta \ell(F(x,\theta),y)$ encodes the maximum amount of feature from the input data $x$, making the gradient space most susceptible to leakage. 
\begin{proof}
\textbf{Continuity of $L(\theta, \phi)$.}
From Assumption (\ref{assumption2}), since $\ell(\cdot,\cdot)$ is continuously differentiable with respect to $\theta$, the gradient $\nabla_\theta \ell(F(x,\theta),y)$ is continuous with respect to $\theta$. The projector $\Pi$ is a fixed-position sparse sampling linear operator, so the composite function $\Pi(\nabla_\theta l(F(x,\theta),y))$ is also continuous with respect to $\theta$. By Assumption \ref{assumption3}, the decoder $\mathcal{D}$ is continuous. Therefore, the composite function $\mathcal{D}(\Pi(\nabla_\theta l(F(x,\theta),y)))$ is continuous with respect to $\theta$ and $\phi$. The squared Euclidean norm $\|\cdot\|_2^2$ is a continuous. Hence, the loss function $L(\theta, \phi)$ is continuous with respect to $\theta$ and $\phi$.

\textbf{Existence of Global Minimum.}
From Assumption (\ref{assumption1}) and Assumption (\ref{assumption3}), the parameter space $\Theta$ and $\Phi$ is a non-empty compact set. By the Weierstrass Extreme Value Theorem \cite{rudin1976principles, bartle2011introduction}, any continuous function on a compact set attains its maximum and minimum values. Therefore, there exists ${\theta}^* \in \Theta$ and ${\phi}^* \in \Phi$ such that $\theta^*, {\phi}^*=\arg\min_{\theta\in\Theta, \phi\in\Phi}L(\theta, \phi)$.

\textbf{Maximum Vulnerability of the Gradient Space.}
At $\theta = \theta^{*}$ and $\phi = \phi^{*}$, the loss function $L(\theta, \phi)$ attains its global minimum, which indicates that the reconstruction error is minimized. This indicates that the encoding and decoding processes of the gradient have reached an optimal state. If a better decoder or gradient construction existed, it would further reduce $L(\theta, \phi)$, contradicting the minimality of $L(\theta^{*}, \phi^{*})$. Therefore, at $\theta = \theta^{*}$ and $\phi = \phi^{*}$, the risk of data leakage from the gradient space to the input data $x$ is maximized. This means the gradient $\nabla_\theta \ell(F(x,\theta^*),y)$ contains the most features of $x$, rendering the gradient space most vulnerable.
\end{proof}

\subsection{Theoretical Guarantee of Stealthiness}
\label{sec:stealthy-proof}

In the previous subsection, we established the existence of optimal poisoned parameters $(\theta^*, \phi^*)$ minimizing the poisoning loss in Equation~\eqref{eq:decoder}. Here we theoretically demonstrate that gradients derived from $\theta^*$ exhibit strong stealthiness, undermining detection methods such as D-SNR~\cite{Garov2024Hiding}.

Stealthiness can be quantified by the variance of gradient norms within a batch. Let $g_i$ be the gradient norm of the $i$-th sample using optimal parameters $\theta^*$:
\begin{equation}
g_i = \left\|\frac{\partial \ell(F(x_i; \theta^*), y_i)}{\partial \theta^*}\right\|_2.
\end{equation}

We define stealthiness as minimizing the variance of $g_i$:
\begin{equation}
\mathrm{Var}(g_i) = \frac{1}{B}\sum_{i=1}^{B}(g_i - \mu_g)^2, \quad \text{with} \quad \mu_g = \frac{1}{B}\sum_{i=1}^{B} g_i.
\end{equation}

\begin{theorem}[Gradient Uniformity and Stealthiness]
Under Assumptions~\ref{assumption1}--\ref{assumption3}, the optimal poisoned parameters $(\theta^*, \phi^*)$ ensure that the gradient variance is bounded by a small constant $\epsilon$, leading to bounded D-SNR indistinguishable from naturally trained gradients:
\begin{equation}
\mathrm{Var}(g_i) \leq \epsilon, \quad D-SNR(\theta^*) \leq \gamma\cdot\epsilon,
\end{equation}
where $\gamma>0$ is a constant dependent on the batch size and model structure.
\end{theorem}

\begin{proof}
At the optimal $(\theta^*,\phi^*)$, the reconstruction loss (Equation~\eqref{eq:decoder}) attains its minimum, ensuring uniform reconstruction errors across the batch:
\begin{equation}
\|x_i - x'_i\|_2 \leq \delta, \quad \forall i,
\end{equation}
for a small constant $\delta$. Given the Lipschitz continuity \cite{rudin1976principles} of decoder $\mathcal{D}$ with constant $L_D$, we have:
\begin{equation}
\resizebox{\linewidth}{!}{$
\|\Pi(\nabla_{\theta^*}\ell(F(x_i,\theta^*),y_i))-\Pi(\nabla_{\theta^*}\ell(F(x_j,\theta^*),y_j))\|_2 \leq 2L_D^{-1}\delta.
$}
\end{equation}

Since the projector $\Pi$ preserves norm differences up to a constant factor $M_\Pi$, we derive:
\begin{equation}
|g_i - g_j| \leq 2M_\Pi L_D^{-1}\delta, \quad \forall i,j.
\end{equation}

Thus, gradient variance satisfies:
\begin{equation}
\mathrm{Var}(g_i) \leq 4(M_\Pi L_D^{-1})^2\delta^2 := \epsilon.
\end{equation}

As D-SNR monotonically decreases with gradient variance~\cite{Garov2024Hiding}, we conclude:
\begin{equation}
D-SNR(\theta^*)\leq \gamma\cdot\epsilon,
\end{equation}
establishing that gradients from $\theta^*$ remain indistinguishable from naturally trained models, ensuring stealthiness.
\end{proof}


\section{Experimental Evaluation}
\label{sec:experiments}
In this section, we present a series of experiments to evaluate the effectiveness of the proposed method. The experiment results show that \name significantly outperforms the SOTA AGLAs in reconstruction quality and stealthiness.

\subsection{Setup}
We use the ResNet18 \cite{he2016deep} as the default global model for FL. The CIFAR10, CIFAR100 \cite{krizhevsky2009learning}, and TinyImageNet \cite{le2015tiny} datasets are employed as the training data for clients. 
The three classic evaluation metrics for reconstruction quality, namely PSNR \cite{hore2010image}, SSIM \cite{zhang2018unreasonable}, and LPIPS \cite{wang2004image}, are employed to assess the attack quality. We use the detection metric D-SNR \cite{Garov2024Hiding} to evaluate the stealthiness of model modifications. We set the default projection ratio to 0.4\% and employ a linear layer as the default structure for the discriminator. We compare our method with the closely related SOTA methods, Fishing \cite{wen2022fishing} and SEER \cite{Garov2024Hiding} with the maximal brightness as the selected property. As our method is the first to poison model parameters for enhanced data leakage across entire batches, we also evaluate its performance against popular naive model initialization methods, including Random, Xavier \cite{glorot2010understanding}, and He \cite{he2015delving}. In our implementation, Xavier initialization uses a uniform distribution to balance variance across layers, while He initialization adapts weights for leaky ReLU activations to ensure smoother gradient flow during training.
\begin{table}[h]
  \centering
  \caption{Comparison of reconstruction performance among Fishing \cite{wen2022fishing}, SEER \cite{Garov2024Hiding}, and \name (Ours) on CIFAR100 with a batch size of 8.}
  \begin{tabularx}{\linewidth}{c|XXX}
    \toprule
          & Min PSNR $\uparrow$ & Pruned Average PSNR $\uparrow$ & Max PSNR $\uparrow$ \\
    \midrule
    Fishing \cite{wen2022fishing}  & 0.00000     & 0.00000     & 12.92526 \\
    SEER \cite{Garov2024Hiding}  & 0.00000     & 0.00000     & 15.97548 \\
    \name (Ours)  & \cellcolor{customblue}\textbf{20.37788} & \cellcolor{customblue}\textbf{21.59001} & \cellcolor{customblue}\textbf{22.86605} \\
    \bottomrule
  \end{tabularx}%
  \label{compareFishSeer}%
\end{table}%
\begin{table*}[b]
  \centering
  \caption{Performance comparison of the proposed \name against baseline model initializations Random, Xavier, and He on CIFAR10, CIFAR100, and TinyImageNet datasets.}
    \begin{tabular}{ccccccccccccc}
    \toprule
    \multirow{2}[4]{*}{Method} & \multirow{2}[4]{*}{Interation} & \multicolumn{3}{c}{CIFAR10} &       & \multicolumn{3}{c}{CIFAR100} &       & \multicolumn{3}{c}{TinyImageNet} \\
\cmidrule{3-5}\cmidrule{7-9}\cmidrule{11-13}          &       & PSNR $\uparrow$ & SSIM $\uparrow$ & LPIPS $\downarrow$ &       & PSNR $\uparrow$ & SSIM $\uparrow$ & LPIPS $\downarrow$ &       & PSNR $\uparrow$ & SSIM $\uparrow$ & LPIPS $\downarrow$ \\
    \midrule
    Random+iDLG & \multirow{4}[2]{*}{200} & 15.86852 & 0.595493 & 0.289589 &       & 16.996 & 0.537246 & 0.344143 &       & 14.04722 & 0.222929 & 0.575494 \\
    Xavier \cite{glorot2010understanding}+iDLG &       & 20.77170  & 0.78643  & 0.24686  &       & 19.85292  & 0.73102  & 0.26831  &       & 12.18539  & 0.23047  & 0.58592  \\
    He \cite{he2015delving}+iDLG &       & -1.15507  & -0.00052  & 0.72385  &       & -1.94603  & -0.00166  & 0.75271  &       & -1.05347  & -0.00043  & 0.80013  \\
    \name(Ours)+iDLG &       & \cellcolor{customblue}\textbf{29.70104}  & \cellcolor{customblue}\textbf{0.86494}  & \cellcolor{customblue}\textbf{0.10815}  &       & \cellcolor{customblue}\textbf{28.44256}  & \cellcolor{customblue}\textbf{0.89706}  & \cellcolor{customblue}\textbf{0.09068}  &       & \cellcolor{customblue}\textbf{19.94374}  & \cellcolor{customblue}\textbf{0.62674}  & \cellcolor{customblue}\textbf{0.22103}  \\
    \midrule
    Random+IG & \multirow{4}[2]{*}{1000} & 19.16213  & 0.62193  & 0.30759  &       & 19.34636  & 0.63830  & 0.31395  &       & 15.47002  & 0.25633  & 0.52080  \\
    Xavier \cite{glorot2010understanding}+IG &       & 24.47016  & 0.86330  & 0.14445  &       & 23.16149  & 0.80442  & 0.17232  &       & 13.06239  & 0.21848  & 0.57367  \\
    He\cite{he2015delving}+IG &       & 13.30730  & 0.10187  & 0.62424  &       & 10.97889  & 0.09065  & 0.66285  &       & 12.70339  & 0.20885  & 0.72560  \\
    \name(Ours)+IG &       & \cellcolor{customblue}\textbf{31.96512}  & \cellcolor{customblue}\textbf{0.91660}  & \cellcolor{customblue}\textbf{0.07358}  &       & \cellcolor{customblue}\textbf{31.55152}  & \cellcolor{customblue}\textbf{0.92673}  & \cellcolor{customblue}\textbf{0.06176}  &       & \cellcolor{customblue}\textbf{28.62325}  & \cellcolor{customblue}\textbf{0.91401}  & \cellcolor{customblue}\textbf{0.07571}  \\
    \bottomrule
    \end{tabular}%
  \label{compare2}%
\end{table*}%

\subsection{Main Results}
\textbf{Comparison between \name and SOTA AGLAs}.
Firstly, we compare the performance of the proposed \name with two SOTA AGLAs, Fishing \cite{wen2022fishing}, and SEER \cite{Garov2024Hiding}, on CIFAR100 with a batch size of 8. Table \ref{compareFishSeer} reports the minimum PSNR, pruned average PSNR, and maximum PSNR over 100 batches, where a PSNR of 0 indicates no reconstruction. \name achieves significantly higher PSNR, consistently reconstructing all samples per batch with minimal variation, while SOTA methods reconstruct only one sample per batch. 
The visual comparison in Figure~\ref{visualresults} demonstrates the superiority of the proposed EGGV method. Specifically, EGGV successfully reconstructs every sample in the batch with a high similarity to the originals. In contrast, SOTA methods (Fishing and SEER) reconstruct only a single image, and the similarity between these reconstructions and the originals is significantly lower than that of EGGV.
\begin{remark} 
Our method applies not only to model initialization but to any round in FL. Unlike prior work~\cite{wen2022fishing, zhang2022compromise} producing suspicious parameters (e.g., zeros or ones), our poisoned parameters exhibit natural distributions. Moreover, since each training round aggregates updates from multiple clients, individual clients remain unaware of the aggregated parameter state, enabling our attack to stealthily poison model parameters at any training round.
\end{remark}
\begin{figure*}[t]
    \centering
    \includegraphics[scale=0.2]{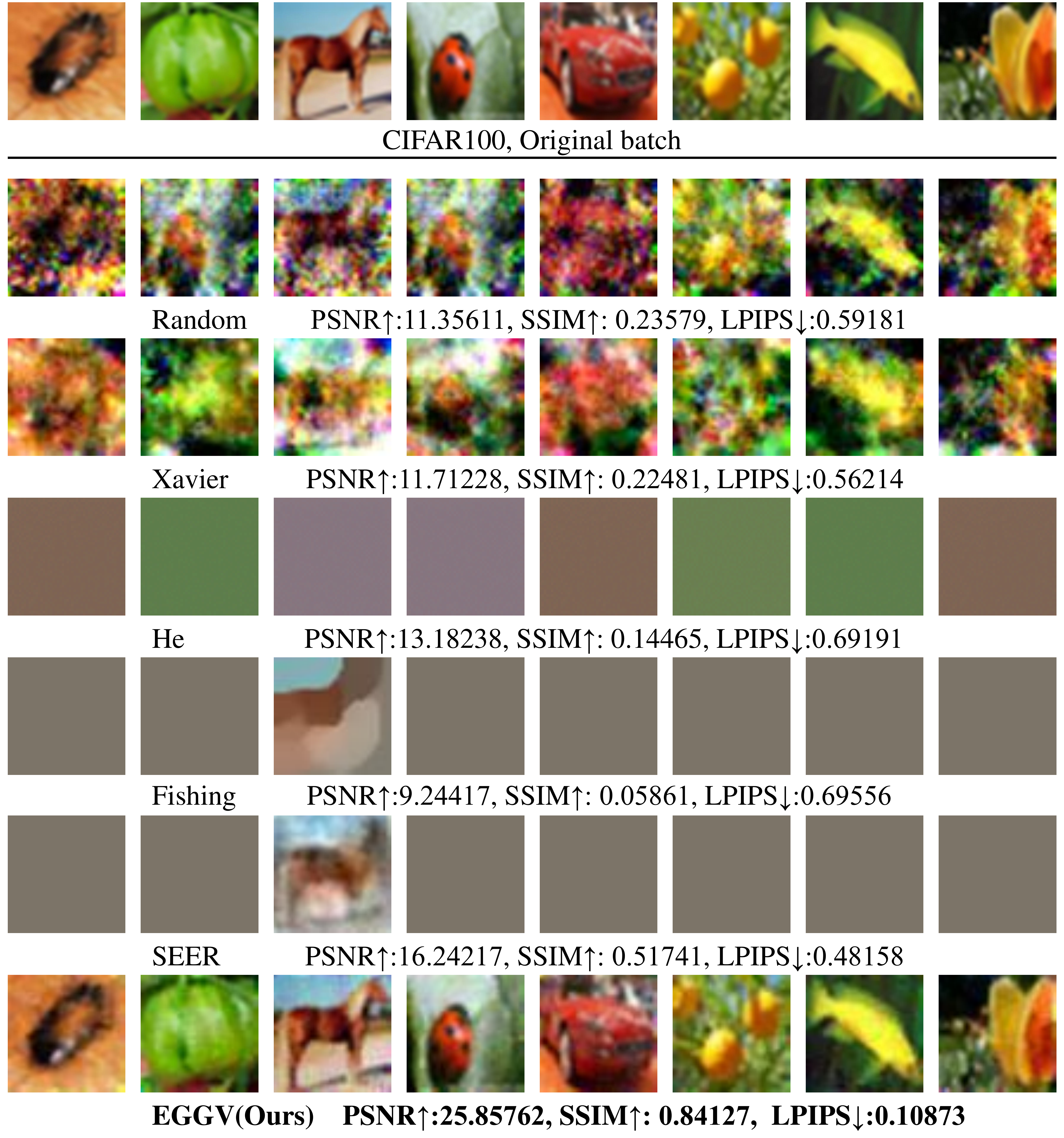}
    \caption{Visual reconstruction of IG on the model with \name poisoning, Xavier initialization, and He initialization.}
    \label{visualresults}
\end{figure*}

\textbf{Comparison between \name and Popular Model Initialization Methods in Enhancing PGLAs}. Considering that \name is the first AGLA to reconstruct all samples in the batch, we compare its performance with three naive model initialization methods. We implement the iDLG \cite{zhao2020idlg} and IG \cite{geiping2020inverting} on models that proposed \name poisons, naively initialized by Random, Xavier, and He. As shown in Table \ref{compare2}, the \name significantly outperforms Random, Xavier, and He initialization methods across all three datasets, regardless of whether iDLG or IG is used. Figure \ref{visualresults} provides a visual comparison of reconstruction results, clearly illustrating the superiority of \name.

He initialization is widely recognized for its advantages in global model training when used by honest servers, but it often results in attack failures for adversaries, including the server itself. This indicates that relying solely on original model parameters results in poor attack performance. Our research further shows that to improve the effectiveness of attacks, adversaries cannot rely only on standard model parameters. Instead, they should adopt poisoning techniques like \name to actively manipulate the gradient space.

The experimental results also highlight the critical importance of the gradient position within the gradient space for the success of GLAs. Unfortunately, previous PGLAs have overlooked this factor. Traditional PGLAs are typically limited by the current state of the model parameters, thus making it difficult to achieve optimal results. Although AGLAs attempt to address this by poisoning model parameters, their effectiveness is limited to a small subset of samples within the batch, as illustrated in Figure \ref{visualresults}. Notably, \name is the first method to tackle this key challenge for both PGLAs and AGLAs by poisoning model parameters to enhance the gradient vulnerability across the entire batch.

\textbf{Stealthiness Comparisons of \name with SOTA AGLAs}.
We calculate 100 gradients on CIFAR100 with each ResNet18 poisoned by Fishing, SEER, and \name and initialized by naive initialization methods Random, Xavier, and He. We then report the D-SNR values for these gradients. As illustrated in Figure \ref{dsnridsnr}, \name demonstrates high stealthiness, achieving D-SNR values similar to those of the naive initialization methods Random, Xavier, and He. In contrast, Fishing and SEER show significantly higher D-SNR values, suggesting that clients can detect these methods more easily. This is because \name evenly enhances the leakage potential of all samples without introducing any gradient bias. In contrast, SOTA methods exhibit biased gradients across all layers, making them more prone to detection.
\begin{figure}[t]
    \centering
    \vspace{0.0cm}  
    \includegraphics[scale=0.14]{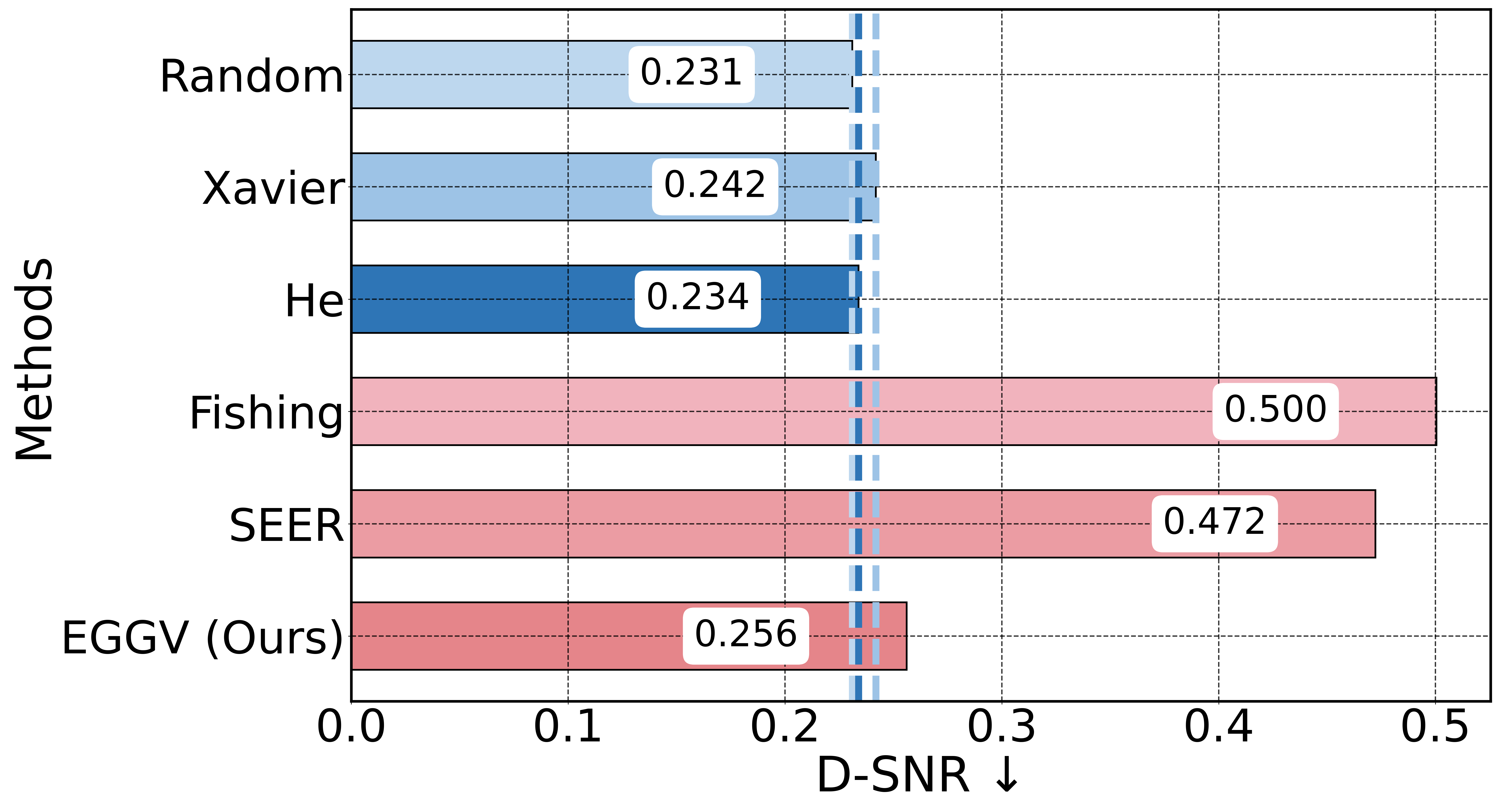}
    \caption{Bar chart of D-SNR value of gradients generated by models with three naive initialization methods (Random, Xavier, He) and three poisoning methods (Fishing, SEER, \name) on 100 same batches. Lower values indicate greater stealthiness. \name achieves a highly stealthy, closely approaching D-SNR of standard initialization methods.}
    \label{dsnridsnr}
    \vspace{-0.3cm}
\end{figure}

\subsection{Ablation Study}
\textbf{Evaluating Gradient Space Vulnerability Using the Discriminator Instead of End-to-End Iterative Attacks}. We now turn to explore the effectiveness of using a discriminator to assess gradient space vulnerability, as opposed to traditional end-to-end iterative attacks. We randomly select two model directions, $x$ and $y$, in the model parameters space and systematically shift the poisoned model parameter $\theta^{*}$ along these axes, generating 441 model parameters. The discriminator evaluates the gradient vulnerability for each of these parameters, and the resulting contour map of gradient vulnerability is shown in Figure \ref{results2}(a). In this map, we select four points: $\theta_{1}$, $\theta_{2}$, $\theta_{3}$, $\theta^{*}$, with corresponding vulnerability scores of 14.99276, 5.42985, 1.28999, and 0.00655 assigned by the discriminator. Subsequently, we perform the IG attack on these four models with the CIFAR10 dataset. Figure \ref{results2}(b) depicts the PSNR convergence during the attacks on these four models. As expected, $\theta^{*}$ yields the best reconstruction. The reconstructed images are highly similar to the original input, and the PSNR value remains the highest throughout the convergence process, significantly outperforming the other three parameters. These findings demonstrate the discriminator’s effectiveness in evaluating gradient space vulnerability and predicting the likelihood of a successful reconstruction attack. In contrast to traditional end-to-end reconstruction methods, this method enables attackers to quickly identify and poison model parameters that could lead to attack failure, thereby improving the overall success rate for attacks.
\begin{figure*}[b]
    \centering
    \includegraphics[scale=0.155]{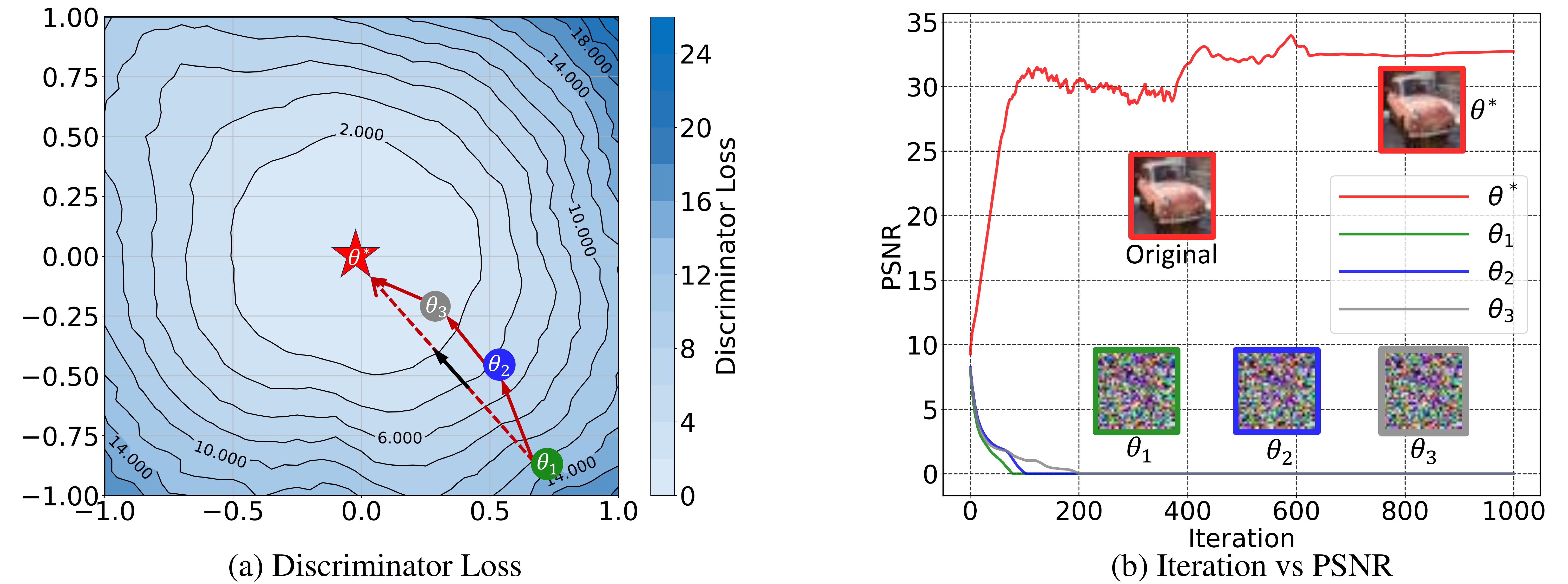}
    \caption{Figure (a): Contour map of the discriminator loss landscape across 441 model parameters generated by shifting $\theta^{*}$ along two random directions. The map marks the selected points $\theta_{1}$, $\theta_{2}$, $\theta_{3}$, and $\theta^{*}$. Figure (b): PSNR curves from IG attacks on a ResNet18 model at $\theta_{1}$, $\theta_{2}$, $\theta_{3}$, and $\theta^{*}$, demonstrating that $\theta^{*}$ achieves the best reconstruction quality.}
    \label{results2}
\end{figure*}

\textbf{Exploring the Relationship Between Gradient Space Vulnerability and Model Accuracy}.
To explore the relationship between gradient vulnerability and model accuracy, we conduct experiments by shifting the poisoned model parameter $\theta^{*}$ evenly 21 times along two randomly selected directions, $x$ and $y$, generating 441 model parameters. Each parameter receives a gradient vulnerability score from the discriminator, visualized in the 3D surface plot at the top of Figure \ref{dsnridsnr}, representing the gradient vulnerability landscape across the parameters space. We then evaluate the classification accuracy of these same 441 model parameters using the CIFAR10 dataset, producing the lower plot of Figure \ref{dsnridsnr}. This plot shows the model accuracy at the same parameter positions as in the gradient vulnerability plot. A comparison between the two plots reveals that the model parameters with the highest gradient vulnerability do not coincide with those that yield the highest accuracy. In fact, model parameters with the greatest gradient vulnerability often show low accuracy, indicating no direct correlation between gradient vulnerability and model accuracy.
\begin{figure*}[t]
    \centering
    \includegraphics[scale=0.135]{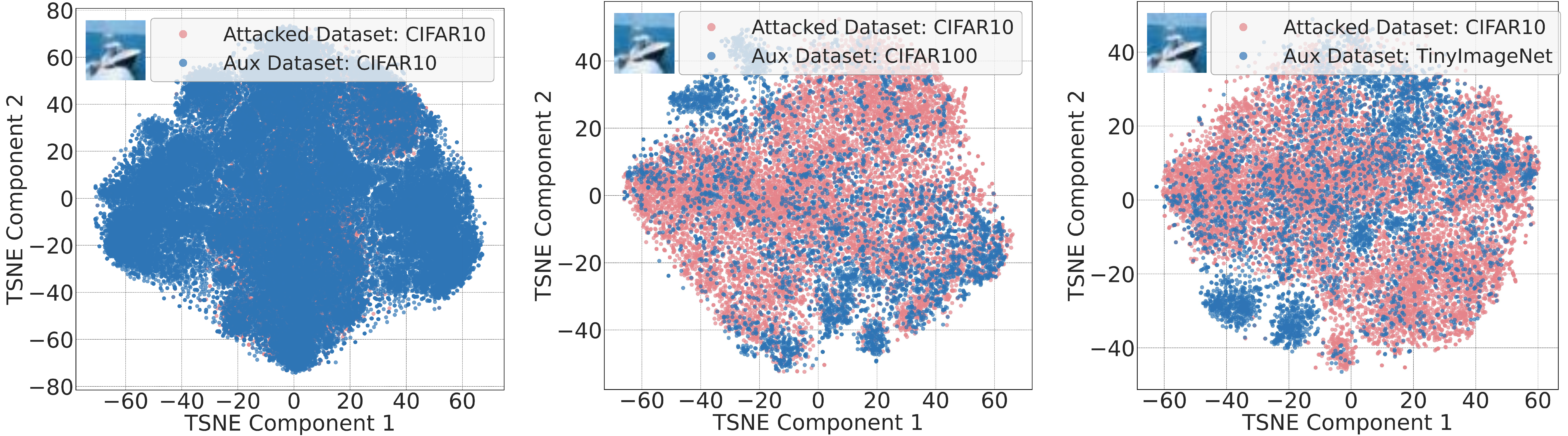}
    \caption{Visualization of distributional differences between auxiliary and target datasets and IG reconstruction results on \name-poisoned models using each auxiliary dataset.}
    \label{diffaux}
\end{figure*}
\begin{figure}[h]
    \centering
    \includegraphics[scale=0.14]{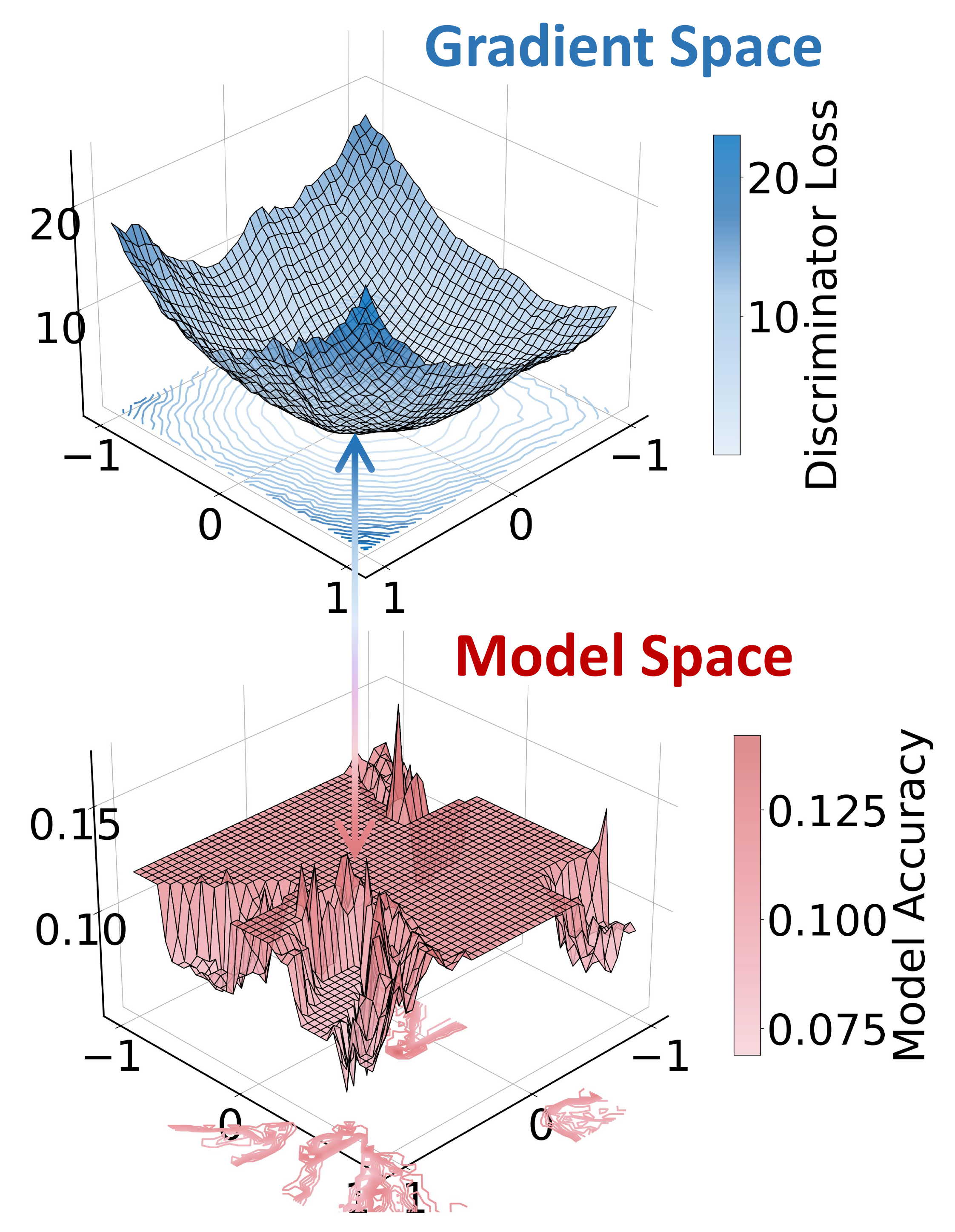}
    \caption{3D surface plots of the discriminator loss (top) and model accuracy (bottom) for the same 441 model parameters shifted $\theta^{*}$ along two random directions, illustrating the relationship between gradient vulnerability and model accuracy across the parameters space.}
    \label{model_gradient_space}
\end{figure}
\begin{table}[t]
  \centering
  \caption{Reconstruction results of IG attack ResNet18 poisoned by \name with projection ratios of 1.6\%, 0.8\%, and 0.4\%. \name at 0.4\% projection ratio achieves the best overall performance.}
  \resizebox{\linewidth}{!}{ 
    \begin{tabularx}{\linewidth}{c|XXX}
    \toprule
    Initiation Method & \multicolumn{1}{c}{PSNR $\uparrow$} & \multicolumn{1}{c}{SSIM $\uparrow$} & \multicolumn{1}{c}{LPIPS $\downarrow$} \\
    \midrule
    Random + IG & 15.47002  & 0.25633  & 0.52080 \\
    Xavier \cite{glorot2010understanding} + IG & 13.06239  & 0.21848  & 0.57367 \\
    He \cite{he2015delving} + IG & 12.70339  & 0.20885  & 0.72560 \\
    \name (Ours) ($\rho: 1.60\%$) + IG & 25.24615 & 0.80658 & 0.12803 \\
    \name (Ours) ($\rho: 0.80\%$) + IG & 27.63074 & 0.85759 & 0.09526 \\
    \name (Ours) ($\rho: 0.40\%$) + IG & \cellcolor{customblue}\textbf{28.62325} & \cellcolor{customblue}\textbf{0.91402} & \cellcolor{customblue}\textbf{0.07575} \\
    \bottomrule
    \end{tabularx}
  } 
  \label{difprojectratio}%
\end{table}

\textbf{The Effect of Different Gradient Projection Ratios}.
A crucial component of the \name is the projector, which compresses high-dimensional gradients into a one-dimensional vector. Next, we examine how different projection ratios influence \name to enhance the vulnerability of the gradient space. We select commonly used Random, Xavier, and He initialization methods as comparison benchmarks, and set three different gradient projection ratios of 1.60\%, 0.80\%, and 0.40\%. IG is used to conduct gradient leakage on the TinyImageNet dataset with the models initialized by the Random and \name. Comparison experimental results in Table \ref{difprojectratio} show that \name consistently enhances the gradient vulnerability across all three projection ratios, outperforming Random, Xavier, and He initialization methods. We observe an interesting phenomenon: the smallest projection ratio of 0.4\% achieves the best effect. This phenomenon can be attributed to the fact that smaller projection ratios force the model to embed more data features into the entire gradient, ensuring that the projected gradients retain enough data features to be inverted by the discriminator back to the original input. Under a smaller projection ratio, the model will more actively adjust its own parameters, thus containing more data features in the entire gradient, which is more conducive to the attack of subsequent attack methods.

\textbf{Exploring the Effect of Distribution Differences Between Auxiliary and Target Datasets on \name}.
Next, we explore the effect of distributional differences between auxiliary and target datasets on the performance of the \name. We select CIFAR10 as the target dataset and CIFAR10, CIFAR100, and TinyImageNet as the auxiliary datasets, respectively. To align the class counts with CIFAR10 to ensure compatibility for model poisoning, 10 classes are randomly sampled from CIFAR100 and TinyImageNet to serve as auxiliary datasets. Table \ref{diffauxtable} reports the performance of iDLG and IG attacks on \name-poisoned models under the above setting. The results show that \name achieves comparable PSNR values across CIFAR10, CIFAR100, and TinyImageNet auxiliary datasets, highlighting its robustness to distributional differences between auxiliary and target datasets. In addition, we use the t-SNE algorithm \cite{van2008visualizing} to reduce the dimensionality of these datasets to 2D and visualize the data distribution of the auxiliary and target datasets, along with the corresponding reconstruction results shown in Figure \ref{diffaux}. The visualizations confirm that the effectiveness of the \name attack is independent of distributional differences between the auxiliary and target datasets. In contrast, some SOTA methods, such as SEER, require the auxiliary dataset to be the training dataset of the target dataset, which limits their application to practical FL systems.
\begin{table}[t]
  \centering
  \footnotesize 
  \vspace{0.0cm}
  \caption{Performance of iDLG and IG attacks on \name-poisoned models with CIFAR10 as target and different auxiliary datasets, showing \name's robustness to distributional shifts.}
  \vspace{-0.2cm}
    \begin{tabularx}{\linewidth}{cXXX}
    \toprule
    \multirow{2}[3]{*}{\parbox[t]{2.3cm}{\centering Attacked Dataset: CIFAR10}} & \multicolumn{3}{c}{Auxiliary Datasets (PSNR $\uparrow$)} \\
    \cmidrule{2-4}          & \centering CIFAR10 & CIFAR100 & TinyImageNet \\
    \midrule
    \name + iDLG & 29.70104  & 29.93789  & 31.83867  \\
    \name + IG & 31.96511  & 31.28306  & 31.64095  \\
    \bottomrule
    \end{tabularx}%
  \label{diffauxtable}%
\end{table}%

\section{Conclusion}
\label{sec:conclusion}

In this work, we introduce a new backdoor-theoretic perspective to rethink and frame existing AGLAs. Through this lens, we theoretically identify that all prior AGLAs suffer from incomplete attack coverage and detectability issues. We further propose \name, a new solution that extends existing AGLAs to be more comprehensive and stealthier to address the above challenges. \name is the first AGLA capable of fully inverting all samples within a target batch while evading existing detection metrics. Extensive experiments demonstrate that EGGV significantly outperforms SOTA AGLAs in both stealthiness and attack coverage. These results encourage the FL community to explore further privacy protection mechanisms to counter these emerging security risks.


\bibliographystyle{IEEEtran}
\bibliography{reference.bib}

\begin{thebibliography}{10}
\providecommand{\url}[1]{#1}
\csname url@samestyle\endcsname
\providecommand{\newblock}{\relax}
\providecommand{\bibinfo}[2]{#2}
\providecommand{\BIBentrySTDinterwordspacing}{\spaceskip=0pt\relax}
\providecommand{\BIBentryALTinterwordstretchfactor}{4}
\providecommand{\BIBentryALTinterwordspacing}{\spaceskip=\fontdimen2\font plus
\BIBentryALTinterwordstretchfactor\fontdimen3\font minus \fontdimen4\font\relax}
\providecommand{\BIBforeignlanguage}[2]{{%
\expandafter\ifx\csname l@#1\endcsname\relax
\typeout{** WARNING: IEEEtran.bst: No hyphenation pattern has been}%
\typeout{** loaded for the language `#1'. Using the pattern for}%
\typeout{** the default language instead.}%
\else
\language=\csname l@#1\endcsname
\fi
#2}}
\providecommand{\BIBdecl}{\relax}
\BIBdecl

\bibitem{mcmahan2017communication}
B.~McMahan, E.~Moore, D.~Ramage, S.~Hampson, and B.~A. y~Arcas, ``Communication-efficient learning of deep networks from decentralized data,'' in \emph{Artificial intelligence and statistics}.\hskip 1em plus 0.5em minus 0.4em\relax PMLR, 2017, pp. 1273--1282.

\bibitem{Towardsk}
\BIBentryALTinterwordspacing
K.~A. Bonawitz, H.~Eichner, W.~Grieskamp, D.~Huba, A.~Ingerman, V.~Ivanov, C.~M. Kiddon, J.~Konečný, S.~Mazzocchi, B.~McMahan, T.~V. Overveldt, D.~Petrou, D.~Ramage, and J.~Roselander, ``Towards federated learning at scale: System design,'' in \emph{SysML 2019}, 2019, to appear. [Online]. Available: \url{https://arxiv.org/abs/1902.01046}
\BIBentrySTDinterwordspacing

\bibitem{chilimbi2014project}
T.~Chilimbi, Y.~Suzue, J.~Apacible, and K.~Kalyanaraman, ``Project adam: Building an efficient and scalable deep learning training system,'' in \emph{11th USENIX symposium on operating systems design and implementation (OSDI 14)}, 2014, pp. 571--582.

\bibitem{zhu2019deep}
L.~Zhu, Z.~Liu, and S.~Han, ``Deep leakage from gradients,'' \emph{Advances in neural information processing systems}, vol.~32, 2019.

\bibitem{nowak2024qbi}
M.~V. Nowak, T.~P. Bott, D.~Khachaturov, F.~Puppe, A.~Krenzer, and A.~Hekalo, ``Qbi: Quantile-based bias initialization for efficient private data reconstruction in federated learning,'' \emph{arXiv preprint arXiv:2406.18745}, 2024.

\bibitem{RGAP}
J.~Zhu and M.~B. Blaschko, ``R-gap: Recursive gradient attack on privacy,'' in \emph{International Conference on Learning Representations}, 2021.

\bibitem{yang2022using}
H.~Yang, M.~Ge, K.~Xiang, and J.~Li, ``Using highly compressed gradients in federated learning for data reconstruction attacks,'' \emph{IEEE Transactions on Information Forensics and Security}, vol.~18, pp. 818--830, 2022.

\bibitem{285479}
\BIBentryALTinterwordspacing
K.~Yue, R.~Jin, C.-W. Wong, D.~Baron, and H.~Dai, ``Gradient obfuscation gives a false sense of security in federated learning,'' in \emph{32nd USENIX Security Symposium (USENIX Security 23)}.\hskip 1em plus 0.5em minus 0.4em\relax Anaheim, CA: USENIX Association, Aug. 2023, pp. 6381--6398. [Online]. Available: \url{https://www.usenix.org/conference/usenixsecurity23/presentation/yue}
\BIBentrySTDinterwordspacing

\bibitem{boenisch2023curious}
F.~Boenisch, A.~Dziedzic, R.~Schuster, A.~S. Shamsabadi, I.~Shumailov, and N.~Papernot, ``When the curious abandon honesty: Federated learning is not private,'' in \emph{2023 IEEE 8th European Symposium on Security and Privacy (EuroS\&P)}.\hskip 1em plus 0.5em minus 0.4em\relax IEEE Computer Society, 2023, pp. 175--199.

\bibitem{zhao2023loki}
J.~C. Zhao, A.~Sharma, A.~R. Elkordy, Y.~H. Ezzeldin, S.~Avestimehr, and S.~Bagchi, ``Loki: Large-scale data reconstruction attack against federated learning through model manipulation,'' in \emph{2024 IEEE Symposium on Security and Privacy (SP)}.\hskip 1em plus 0.5em minus 0.4em\relax IEEE Computer Society, 2023, pp. 30--30.

\bibitem{fowl2021robbing}
L.~H. Fowl, J.~Geiping, W.~Czaja, M.~Goldblum, and T.~Goldstein, ``Robbing the fed: Directly obtaining private data in federated learning with modified models,'' in \emph{International Conference on Learning Representations}, 2021.

\bibitem{zhao2020idlg}
B.~Zhao, K.~R. Mopuri, and H.~Bilen, ``idlg: Improved deep leakage from gradients,'' \emph{arXiv preprint arXiv:2001.02610}, 2020.

\bibitem{geiping2020inverting}
J.~Geiping, H.~Bauermeister, H.~Dr{\"o}ge, and M.~Moeller, ``Inverting gradients-how easy is it to break privacy in federated learning?'' \emph{Advances in neural information processing systems}, vol.~33, pp. 16\,937--16\,947, 2020.

\bibitem{yin2021see}
H.~Yin, A.~Mallya, A.~Vahdat, J.~M. Alvarez, J.~Kautz, and P.~Molchanov, ``See through gradients: Image batch recovery via gradinversion,'' in \emph{Proceedings of the IEEE/CVF conference on computer vision and pattern recognition}, 2021, pp. 16\,337--16\,346.

\bibitem{du2024sok}
J.~Du, J.~Hu, Z.~Wang, P.~Sun, N.~Z. Gong, and K.~Ren, ``Sok: Gradient leakage in federated learning,'' \emph{arXiv preprint arXiv:2404.05403}, 2024.

\bibitem{wen2022fishing}
Y.~Wen, J.~A. Geiping, L.~Fowl, M.~Goldblum, and T.~Goldstein, ``Fishing for user data in large-batch federated learning via gradient magnification,'' in \emph{International Conference on Machine Learning}.\hskip 1em plus 0.5em minus 0.4em\relax PMLR, 2022, pp. 23\,668--23\,684.

\bibitem{Garov2024Hiding}
K.~Garov, D.~I. Dimitrov, N.~Jovanovi{\'c}, and M.~Vechev, ``Hiding in plain sight: Disguising data stealing attacks in federated learning,'' in \emph{International Conference on Learning Representations}, 2024.

\bibitem{lecun1998gradient}
Y.~LeCun, L.~Bottou, Y.~Bengio, and P.~Haffner, ``Gradient-based learning applied to document recognition,'' \emph{Proceedings of the IEEE}, vol.~86, no.~11, pp. 2278--2324, 1998.

\bibitem{wainakh2022user}
A.~Wainakh, F.~Ventola, T.~M{\"u}{\ss}ig, J.~Keim, C.~G. Cordero, E.~Zimmer, T.~Grube, K.~Kersting, and M.~M{\"u}hlh{\"a}user, ``User-level label leakage from gradients in federated learning,'' \emph{Proceedings on Privacy Enhancing Technologies}, vol. 2022, no.~2, pp. 227--244.

\bibitem{ma2023instance}
K.~Ma, Y.~Sun, J.~Cui, D.~Li, Z.~Guan, and J.~Liu, ``Instance-wise batch label restoration via gradients in federated learning,'' in \emph{The Eleventh International Conference on Learning Representations}, 2023.

\bibitem{wangtowards}
Y.~Wang, J.~Liang, and R.~He, ``Towards eliminating hard label constraints in gradient inversion attacks,'' in \emph{The Twelfth International Conference on Learning Representations}.

\bibitem{jeon2021gradient}
J.~Jeon, K.~Lee, S.~Oh, J.~Ok \emph{et~al.}, ``Gradient inversion with generative image prior,'' \emph{Advances in neural information processing systems}, vol.~34, pp. 29\,898--29\,908, 2021.

\bibitem{li2022auditing}
Z.~Li, J.~Zhang, L.~Liu, and J.~Liu, ``Auditing privacy defenses in federated learning via generative gradient leakage,'' in \emph{Proceedings of the IEEE/CVF Conference on Computer Vision and Pattern Recognition}, 2022, pp. 10\,132--10\,142.

\bibitem{gu2024federated}
H.~Gu, X.~Zhang, J.~Li, H.~Wei, B.~Li, and X.~Huang, ``Federated learning vulnerabilities: Privacy attacks with denoising diffusion probabilistic models,'' in \emph{Proceedings of the ACM on Web Conference 2024}, 2024, pp. 1149--1157.

\bibitem{zhang2022compromise}
S.~Zhang, J.~Huang, Z.~Zhang, and C.~Qi, ``Compromise privacy in large-batch federated learning via malicious model parameters,'' in \emph{International Conference on Algorithms and Architectures for Parallel Processing}.\hskip 1em plus 0.5em minus 0.4em\relax Springer, 2022, pp. 63--80.

\bibitem{pasquini2022eluding}
D.~Pasquini, D.~Francati, and G.~Ateniese, ``Eluding secure aggregation in federated learning via model inconsistency,'' in \emph{Proceedings of the 2022 ACM SIGSAC Conference on Computer and Communications Security}, 2022, pp. 2429--2443.

\bibitem{nair2010rectified}
V.~Nair and G.~E. Hinton, ``Rectified linear units improve restricted boltzmann machines,'' in \emph{Proceedings of the 27th international conference on machine learning (ICML-10)}, 2010, pp. 807--814.

\bibitem{huggingface}
\BIBentryALTinterwordspacing
T.~Wolf, L.~Debut, V.~Sanh, J.~Chaumond, C.~Delangue, A.~Moi, P.~Cistac, T.~Rault, R.~Louf, M.~Funtowicz, J.~Davison, S.~Shleifer, P.~von Platen, C.~Ma, Y.~Jernite, J.~Plu, C.~Xu, T.~Le~Scao, S.~Gugger, M.~Drame, Q.~Lhoest, and A.~Rush, ``Transformers: State-of-the-art natural language processing,'' in \emph{Proceedings of the 2020 Conference on Empirical Methods in Natural Language Processing: System Demonstrations}, Q.~Liu and D.~Schlangen, Eds.\hskip 1em plus 0.5em minus 0.4em\relax Online: Association for Computational Linguistics, Oct. 2020, pp. 38--45. [Online]. Available: \url{https://aclanthology.org/2020.emnlp-demos.6/}
\BIBentrySTDinterwordspacing

\bibitem{kaggle}
Kaggle, ``Kaggle: Your machine learning and data science community,'' \url{https://www.kaggle.com}, 2023.

\bibitem{OpenML}
J.~Vanschoren, J.~van Rijn, B.~Bischl, and L.~Torgo, ``Openml: Networked science in machine learning,'' \emph{ACM SIGKDD Explorations Newsletter}, vol.~15, pp. 49--60, 12 2013.

\bibitem{glorot2010understanding}
X.~Glorot and Y.~Bengio, ``Understanding the difficulty of training deep feedforward neural networks,'' in \emph{Proceedings of the thirteenth international conference on artificial intelligence and statistics}.\hskip 1em plus 0.5em minus 0.4em\relax JMLR Workshop and Conference Proceedings, 2010, pp. 249--256.

\bibitem{he2015delving}
K.~He, X.~Zhang, S.~Ren, and J.~Sun, ``Delving deep into rectifiers: Surpassing human-level performance on imagenet classification,'' in \emph{Proceedings of the IEEE international conference on computer vision}, 2015, pp. 1026--1034.

\bibitem{rudin1976principles}
\BIBentryALTinterwordspacing
W.~Rudin, \emph{Principles of Mathematical Analysis}, ser. International series in pure and applied mathematics.\hskip 1em plus 0.5em minus 0.4em\relax McGraw-Hill, 1976. [Online]. Available: \url{https://books.google.com.sg/books?id=kwqzPAAACAAJ}
\BIBentrySTDinterwordspacing

\bibitem{bartle2011introduction}
\BIBentryALTinterwordspacing
R.~Bartle and D.~Sherbert, \emph{Introduction to Real Analysis}.\hskip 1em plus 0.5em minus 0.4em\relax Wiley, 2011. [Online]. Available: \url{https://books.google.com.sg/books?id=YawbAAAAQBAJ}
\BIBentrySTDinterwordspacing

\bibitem{he2016deep}
K.~He, X.~Zhang, S.~Ren, and J.~Sun, ``Deep residual learning for image recognition,'' in \emph{Proceedings of the IEEE conference on computer vision and pattern recognition}, 2016, pp. 770--778.

\bibitem{krizhevsky2009learning}
A.~Krizhevsky, G.~Hinton \emph{et~al.}, ``Learning multiple layers of features from tiny images,'' 2009.

\bibitem{le2015tiny}
Y.~Le and X.~Yang, ``Tiny imagenet visual recognition challenge,'' \emph{CS 231N}, vol.~7, no.~7, p.~3, 2015.

\bibitem{hore2010image}
A.~Hore and D.~Ziou, ``Image quality metrics: Psnr vs. ssim,'' in \emph{2010 20th international conference on pattern recognition}.\hskip 1em plus 0.5em minus 0.4em\relax IEEE, 2010, pp. 2366--2369.

\bibitem{zhang2018unreasonable}
R.~Zhang, P.~Isola, A.~A. Efros, E.~Shechtman, and O.~Wang, ``The unreasonable effectiveness of deep features as a perceptual metric,'' in \emph{Proceedings of the IEEE conference on computer vision and pattern recognition}, 2018, pp. 586--595.

\bibitem{wang2004image}
Z.~Wang, A.~C. Bovik, H.~R. Sheikh, and E.~P. Simoncelli, ``Image quality assessment: from error visibility to structural similarity,'' \emph{IEEE transactions on image processing}, vol.~13, no.~4, pp. 600--612, 2004.

\bibitem{van2008visualizing}
L.~Van~der Maaten and G.~Hinton, ``Visualizing data using t-sne.'' \emph{Journal of machine learning research}, vol.~9, no.~11, 2008.

\end{thebibliography}

\begin{IEEEbiography}[{\includegraphics[width=1.0in,height=1.33in,clip,keepaspectratio]{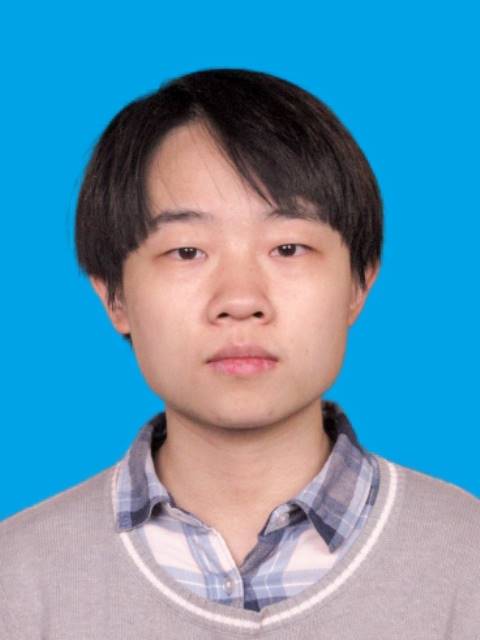}}]{Kunlan Xiang} pursued her M.S. degree from 2022 to 2024 at the School of Computer Science, University of Electronic Science and Technology of China (UESTC), where she is currently working toward the Ph.D. degree. Her research interests include deep learning, AI security, and federated learning security.
\end{IEEEbiography}
\begin{IEEEbiography}[{\includegraphics[width=1in,height=1.33in,clip,keepaspectratio]{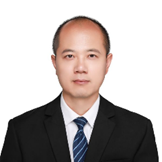}}]{Haomiao Yang}
received the M.S. and Ph.D. degrees in computer applied technology from the University of Electronic Science and Technology of China (UESTC), Chengdu, China, in 2004 and 2008, respectively. He has worked as a Postdoctoral Fellow with the Research Center of Information Cross over Security, Kyungil University, Gyeongsan, South Korea, for one year until June 2013. He is currently a Professor with the School of Computer Science and Engineering and the Center for Cyber Security, UESTC. His research interests include cryptography, cloud security, and cybersecurity for aviation communication.
\end{IEEEbiography}
\begin{IEEEbiography}[{\includegraphics[width=1in,height=1.33in,clip,keepaspectratio]{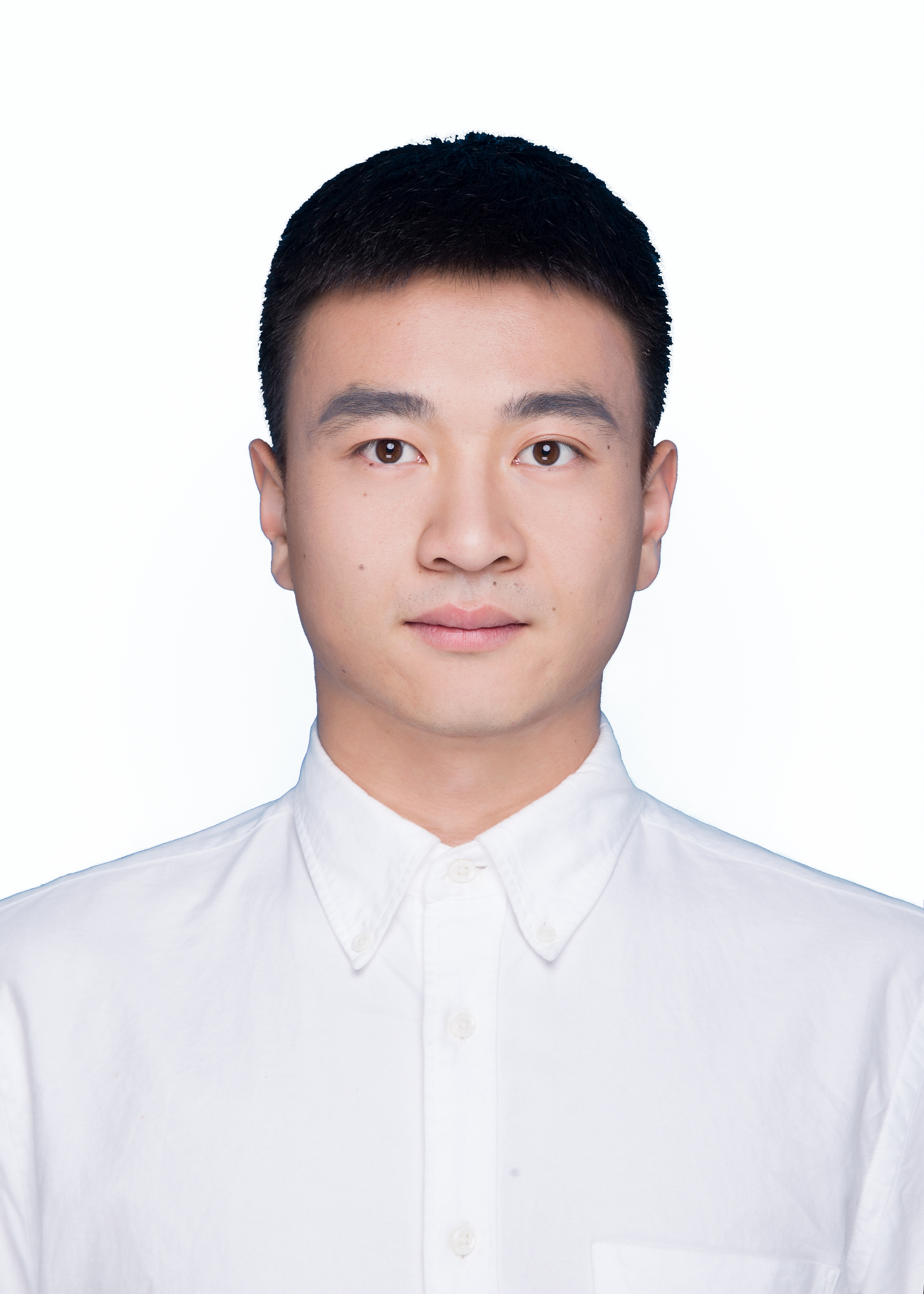}}]{Meng Hao} is currently a Research Scientist at Singapore Management University (SMU). Before that, he obtained his Ph.D degree in 2024 from University of Electronic Science and Technology of China (UESTC). He was also a visiting Ph.D student at Nanyang Technological University. His research interests mainly focus on applied cryptography and machine learning security.
\end{IEEEbiography}
\begin{IEEEbiography}[{\includegraphics[width=1in,height=1.33in,clip,keepaspectratio]{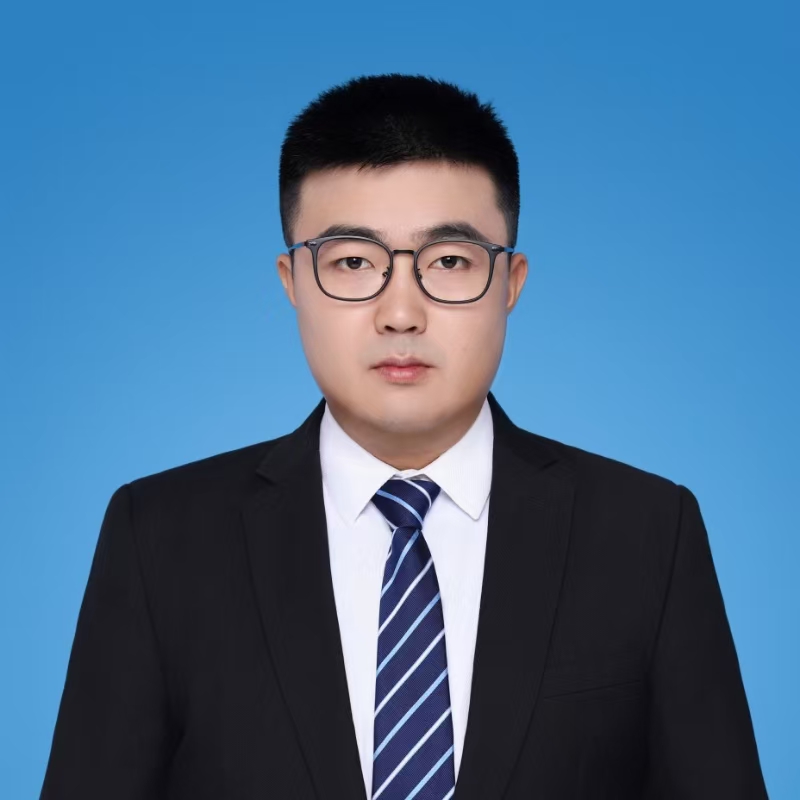}}]{Shaofeng Li} received the Ph.D. degree in the Department of Computer Science and Engineering of Shanghai Jiao Tong University in 2022. He is currently an associate professor at the School of Computer Science and Engineering, Southeast University, Nanjing, China. He focuses primarily on the areas of machine learning and security, specifically exploring the robustness of machine learning models against various adversarial attacks.
\end{IEEEbiography}
\begin{IEEEbiography}[{\includegraphics[width=1in,height=1.33in,clip,keepaspectratio]{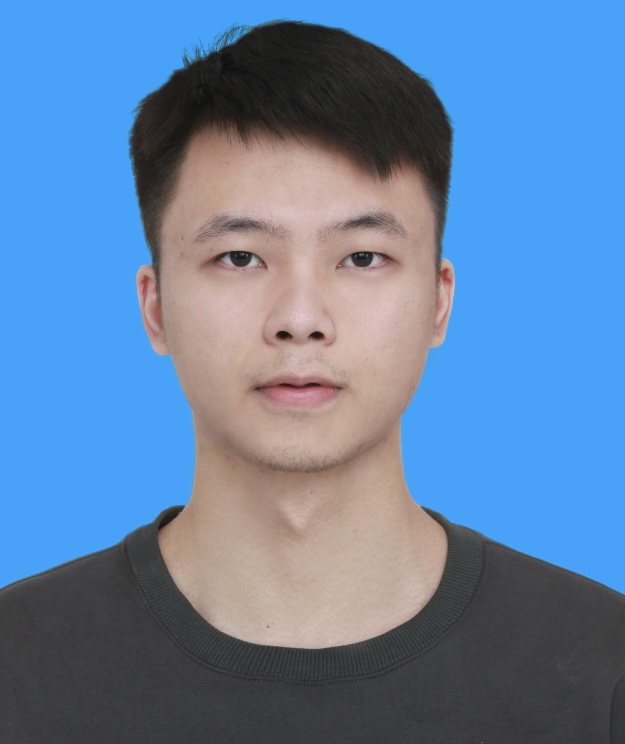}}]{Haoxin Wang} received her B.S. degree the Xi’an University of Technology (XUT) in 2022. He is currently pursuing the M.S. degree in Sichuan University. His research interests include deep learning, artificial intelligence (AI), and time series forecasting.
\end{IEEEbiography}
\begin{IEEEbiography}[{\includegraphics[width=1in,height=1.33in,clip,keepaspectratio]{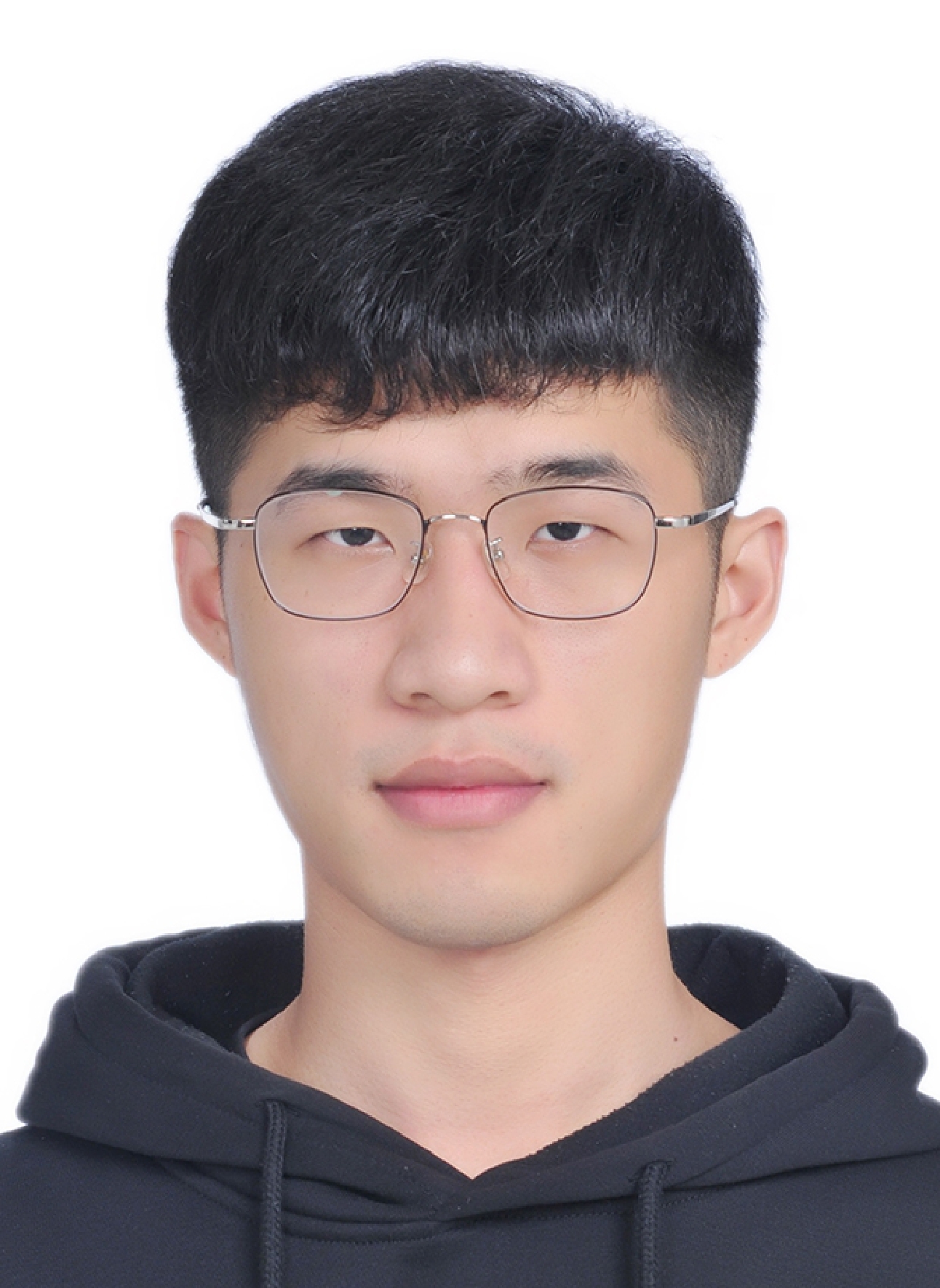}}]{Zikang Ding} received a master's degree in software engineering from East China Normal University. He is currently a Ph.D student at the University of Electronic Science and Technology of China. His current research interests include data privacy and security issues, federated learning, and artificial intelligence security.
\end{IEEEbiography}
\begin{IEEEbiography}[{\includegraphics[width=1in,height=1.33in,clip,keepaspectratio]{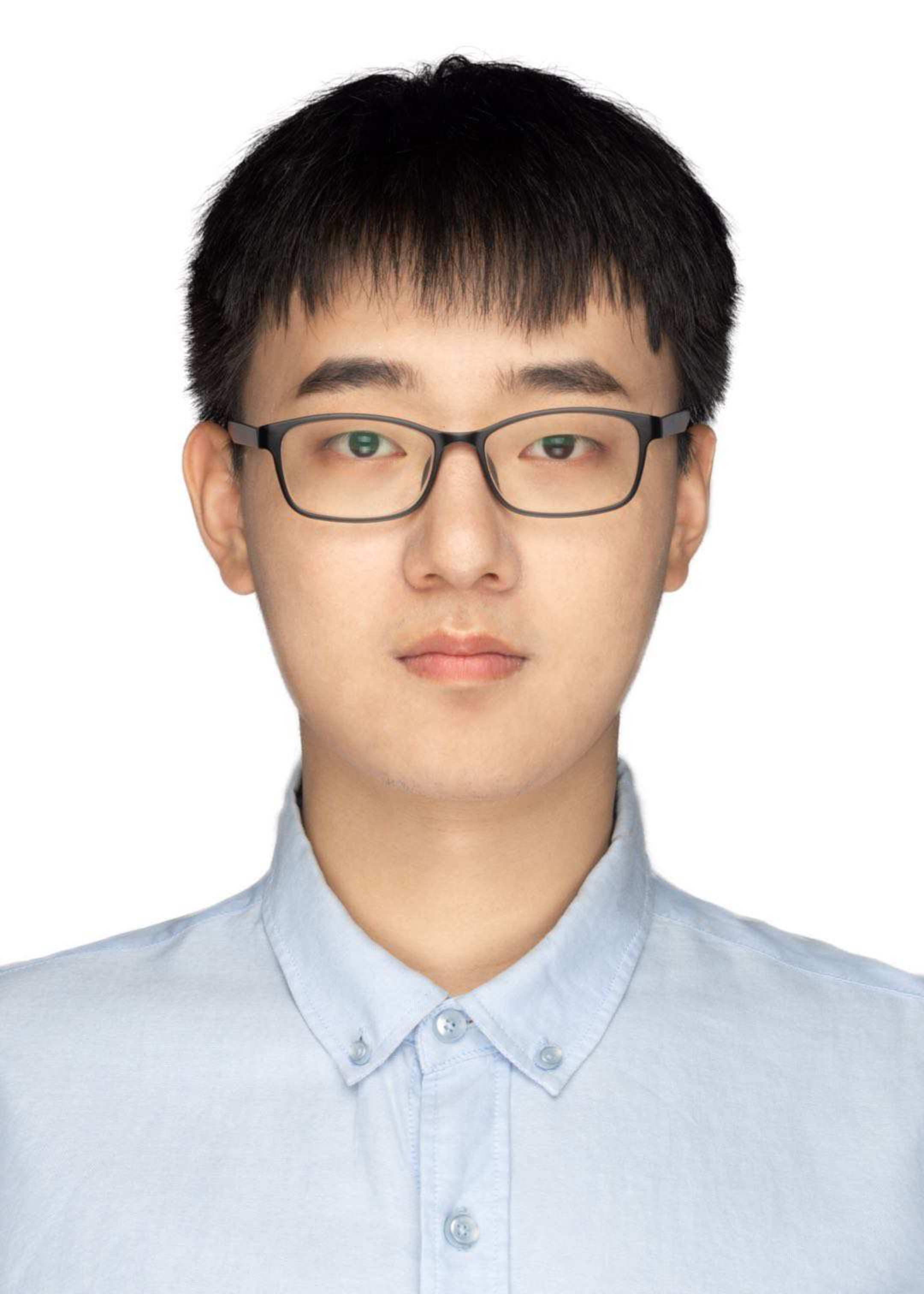}}]{Wenbo Jiang} received the Ph.D. degree in cybersecurity from University of Electronic Science and Technology of China (UESTC) in 2023 and studied as a visiting Ph.D. student from Jul. 2021 to Jul. 2022 at Nanyang Technological University, Singapore. He is currently a Postdoc at UESTC. He has published papers in major conferences/journals, including IEEE CVPR, IEEE TDSC, etc. His research interests include machine learning security and data security.
\end{IEEEbiography}
\begin{IEEEbiography}[{\includegraphics[width=1in,height=1.33in,clip,keepaspectratio]{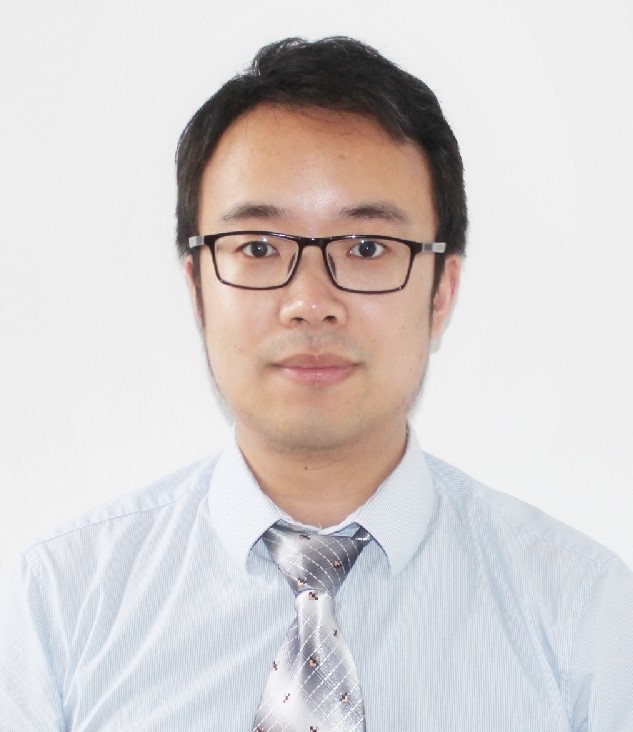}}]{Tianwei Zhang} received the bachelor’s degree from Peking University in 2011, and the PhD degree from Princeton University in 2017. He is an assistant professor with the School of Computer Science and Engineering, Nanyang Technological University. His research focuses on computer system security. He is particularly interested in security threats and defenses in machine learning systems, autonomous systems, computer architecture and distributed systems.
\end{IEEEbiography}

\end{document}